\documentclass[aps,pra,twocolumn,longbibliography,superscriptaddress,10pt,nofootinbib]{revtex4-1}

\usepackage{amsmath, empheq,braket,amsthm,amssymb,graphics,amsfonts,array,xcolor,subfigure, mathrsfs}
\usepackage[unicode=true,pdfusetitle, bookmarks=true,bookmarksnumbered=false,bookmarksopen=false,breaklinks=false,pdfborder={0 0 0},backref=false,colorlinks=true,linkcolor=blue,citecolor=blue,urlcolor=blue]{hyperref}

\newcommand{\bk}[2]{\left< #1 \vphantom{#2} \right|\left. #2 \vphantom{#1} \right>} 

\newcommand{\ha}{\hat{a}}
\newcommand{\ad}{\hat{a}^{\dagger}}

\newcommand{\e}{\mathrm{e}}

\newcommand{\hnl}{\hat{H}_{\mathrm{NL}}}

\newcommand{\eps}{\varepsilon}

\DeclareMathOperator{\sech}{sech}

\newtheorem{theo}{Theorem}
\newtheorem{lem}{Lemma}
\newtheorem{defi}{Definition}

\newtheorem{prob}{Problem}

\begin{document}


\title{Complexity of Gaussian Quantum Optics with a Limited Number of Non-Linearities}


\author{Michael G. Jabbour}
\email{mjabbour@telecom-sudparis.eu}
\thanks{Corresponding author}
\affiliation{SAMOVAR, T\'el\'ecom SudParis, Institut Polytechnique de Paris, 91120 Palaiseau, France}
\affiliation{Centre for Quantum Information and Communication, \'Ecole polytechnique de Bruxelles, CP 165/59, Universit\'e libre de Bruxelles, 1050 Brussels, Belgium}
\affiliation{Department of Physics, Technical University of Denmark, 2800 Kongens Lyngby, Denmark}

\author{Leonardo Novo}
\email{leonardo.novo@inl.int}
\affiliation{International Iberian Nanotechnology Laboratory (INL), Av. Mestre Jos\'e Veiga, 4715-330 Braga, Portugal}
\affiliation{Centre for Quantum Information and Communication, \'Ecole polytechnique de Bruxelles, CP 165/59, Universit\'e libre de Bruxelles, 1050 Brussels, Belgium}


\begin{abstract}
    It is well known in quantum optics that any process involving the preparation of a multimode Gaussian state, followed by a Gaussian operation and Gaussian measurements, can be efficiently simulated by classical computers. Here, we provide evidence that computing transition amplitudes of Gaussian processes with a single-layer of non-linearities is hard for classical computers. To this end, we show how an efficient algorithm to solve this problem could be used to efficiently approximate outcome probabilities of a Gaussian boson sampling experiment. We also extend this complexity result to the problem of computing transition probabilities of Gaussian processes with two layers of non-linearities, by developing a Hadamard test for continuous-variable systems that may be of independent interest. Given recent experimental developments in the implementation of photon-photon interactions, our results may inspire new schemes showing quantum computational advantage or algorithmic applications of non-linear quantum optical systems realizable in the near-term.
\end{abstract}

\maketitle

\section{Introduction.}

Gaussian quantum information processing is a generic term referring to any process involving Gaussian states undergoing Gaussian operations and followed by a Gaussian measurement~\cite{Weedbrook2012}. This particular way of processing quantum states is attractive not only from the theoretical point of view, but also because it is generally easier to realize experimentally. It is known, however, that any amplitude or outcome probability in this framework can be computed efficiently using classical computers ~\cite{sim_Gaussian}. Purely Gaussian quantum processes are thus not interesting for the purpose of achieving quantum computational advantage~\cite{hangleiter2023computational}.

Gaussian boson sampling (GBS) is arguably the most well-known approach for a quantum optical experiment involving Gaussian states that is believed to be hard to simulate using classical computers~\cite{Hamilton2017}. Inspired by the seminal work of Aaronson and Arkhipov~\cite{aaronson2011computational}, it has been shown, using complexity-theoretic arguments, that the existence of an efficient classical algorithm for simulating GBS is unlikely~\cite{Hamilton2017, deshpande2022GBS, grier2022bipartiteGBS}. A Gaussian boson sampler requires the preparation of a multi-mode Gaussian state, resulting from the application of a linear interferometer to a state of many single-mode squeezed states, followed by a photo-counting measurement. The complexity of the problem stems from the measurement in the Fock basis -- a non-Gaussian operation. Recent experiments implementing this protocol have already claimed to have breached the regime of classical simulability~\cite{zhong2020GBS, PhaseprogGBS_2021, madsen2022GBS}. Alternative proposals to achieve quantum advantage require non-Gaussian input states, linear interferometry and Gaussian measurements~\cite{levonBS_Gaussian_meas, chabaud2017photonadded}, or involve the implementation of Instantaneous Quantum Polynomial time (IQP) circuits~\cite{bremner2016IQP} in continuous-variable (CV) systems~\cite{douce2017continuous, douce2019probabilistic}.

An interesting prospect for future boson sampling experiments is to go beyond linear interference~\cite{nonlinear_BS}, motivated by recent experimental developments in the implementation of photon-photon interactions~\cite{exp_kirchmair2013observation, exp_chang2014quantum, exp_hacker2016photon, exp_tiarks2019photon, exp_moreno2021quantum, exp_reuer2022realization}. In particular, the addition of non-linear optical elements to boson sampling schemes may increase the applicability of these devices and their potential to solve useful computational problems. Such interaction terms are also native to other experimental platforms such as cavity QED systems or ultracold atoms where boson sampling experiments can be implemented~\cite{peropadre2016proposal, young2024atomic}.

The problem of analyzing the complexity of classically simulating boson sampling experiments with interaction terms only recently started to be addressed~\cite{maskara2022complexity}. In this work, the authors consider bosonic lattice models with interactions and study complexity transitions between easy and hard to simulate regimes, depending on the interaction strength and evolution time. In this scenario, both the input states, measurements and unitary evolution under an interacting Hamiltonian are non-Gaussian operations. In contrast, in our work, we address the natural question of whether hardness can be achieved when interactions are the \emph{only non-Gaussian operation} applied in a given quantum process, while all other operations, including state preparation and measurements, are Gaussian. More precisely, we consider Gaussian processes with $k$-layers of non-linear photon-photon interactions of the Kerr type, which we denote as GPNL($k$), focusing on the simplest cases with $k=1,2$. A difficulty in proving hardness of classically simulating such processes lies in the fact that the associated amplitudes are not directly related to well-known mathematical objects such as matrix permanents of Hafnians, whose computational complexity is well understood. Nevertheless, we obtain evidence that the problem of approximating up to an exponentially small additive error the transition amplitudes of GPNL(1) and transition probabilities of GPNL(2) is hard for classical computers. We do this by showing explicitly how to connect these problems to the problem of estimating output probabilities of Gaussian boson samplers.  One of our main proof ideas is to show that certain \emph{amplitudes} of GPNL(1) can be post-processed to obtain GBS \emph{probabilities} via a method for quantum eigenvalue estimation~\cite{somma2002, somma2019} (see Fig.~\ref{fig:equivNLBosonSamp}). Moreover, to show the complexity of probabilities of GPNL(2), we prove that they can be used to estimate amplitudes of GPNL(1). For the latter demonstration, we develop a CV version of the so-called Hadamard test~\cite{Cleve1998} involving Gaussian states and Gaussian measurements, which may be of independent interest (see Fig.~\ref{fig:CVHadamardMain}).

\section{Gaussian processes with non-linearities}

Consider an $M$-mode bosonic quantum system, whose $i$\textsuperscript{th} mode is described by an infinite-dimensional Hilbert space with field operators $\{\ha_i,\ad_i\}$. We denote the vacuum state of the system by $\ket{\boldsymbol{0}} = \ket{0}^{\otimes M}$ and define the number operator of each mode $\hat{n}_i = \ad_i \ha_i$. A Gaussian state of the system is one whose so-called Wigner function is a Gaussian probability distribution, while a Gaussian transformation is one that takes a Gaussian state to a Gaussian state~\cite{Weedbrook2012}. For an optical circuit acting on such a system, apart from Gaussian transformations, we consider non-linearities introduced by Hamiltonians of the form
\begin{align}
    H_{\mathrm{SK}} & = \sum_{i=1}^M \eta_i \ad_i \ad_i \ha_i \ha_i ~~\text{(self-Kerr interactions)}  \label{eq:selfK}, \\
    H_{\mathrm{CK}} & = \sum_{i,j=1}^M J_{ij} \ad_i \ha_i \ad_j  \ha_j ~~\text{(cross-Kerr interactions)}. \label{eq:crossK}
\end{align}
More concretely, we are interested in understanding the complexity of computing amplitudes or probabilities (strong simulation) of Gaussian processes with a fixed number $k$ of layers of non-linearities, or GPNL$(k)$  for short, which we define in the following way. 
\begin{defi}[GPNL$(k)$] 
A Gaussian process with $k$ layers of non-linearities is a quantum process involving the preparation of quantum state 
\begin{equation}
    \prod_{j=1}^k \exp(i H_{NL}^{(j)}\Delta t) U_G^{(j)}
    \ket{0}^{\otimes M}
\end{equation}
followed by a Gaussian measurement. Here, $U_G^{(j)}$ are Gaussian unitaries, $\Delta t$ is some constant evolution time  and $H_{NL}^{(j)}$ is either a self-Kerr or a cross-Kerr interaction Hamiltonian from Eqs.~\eqref{eq:selfK} and~\eqref{eq:crossK}, respectively.
\end{defi}
\noindent To argue for hardness of estimating amplitudes or probabilities of such processes, we connect this problem to that of estimating GBS probabilities, whose complexity is better understood. 

\subsection{Complexity of GBS}

Let us recall the typical set-up for a GBS experiment~\cite{Hamilton2017}, briefly reviewing some results about its complexity. Denote a squeezed state with squeezing parameter $r \in \mathbb{R}$ as $\ket{\psi_{\mathrm{SMS}}(r)} = \hat{S}_r \ket{0}$, where $\hat{S}_r = e^{r(a^2-{a^{\dagger}}^2)/2}$ is a single-mode squeezer. The input state to the GBS experiment contains $K$ squeezed states and vacua in the rest of the $M$ modes, $\ket{\Psi_{in}}=\ket{\psi_{\mathrm{SMS}}(r)}^{\otimes K }\ket{0}^{\otimes (M-K)}$.
This state passes through a linear interferometer $\hat{U}$ and the photon number distribution at the output is observed via photon number resolving detectors. The probability of observing an outcome $\ket{S}=\ket{s_1, ..., s_M}$, where each $\ket{s_i}$ is a Fock state, is given by
\begin{equation}\label{eq:P_S}
P_S=|\bra{S} \hat{U}\ket{\Psi_{in}}|^2.
\end{equation}
For certain regimes of the parameters $K$, $r$ and $M$, there is compelling evidence that the computation of such probabilities exactly or with a high precision is a hard problem~\cite{Hamilton2017, kruse2019detailed, deshpande2022GBS}. This complexity stems from the fact that such probabilities are proportional to the modulus square of matrix Hafnians. Hafnians, similiarly to matrix permanents that appear in standard boson sampling, are believed to be difficult to compute (precisely, they are \#P-hard), with the best known classical algorithms for their exact computation running in exponential time.
In fact, a recent proposal for a GBS set-up, named bipartite GBS, shows that it is possible to encode permanents of arbitrary matrices in GBS amplitudes, using properties of the Hafnian~\cite{grier2022bipartiteGBS}.

For the purposes of this work, we focus on the complexity of the problem of estimating GBS probabilities up to an exponentially small additive error, which can be stated as follows.
\begin{prob}[Approx. GBS probabilities]\label{prob:GBS}
Consider a linear interferometer $\hat{U}$ described by an $M\times M$ Haar random unitary $U$ and a number of input squeezed states $K$ with squeezing parameter $r$ such that $\bar{N}= o(\sqrt{M})$, with $\bar{N}= K \sinh^2(r)$. Moreover, consider a collision free outcome $S^*=(s_1, ..., s_M)$, with $s_i\in \{0,1\}$ and total number of photons $N=\sum_i s_i$ satisfying $N<K<M$. The problem is to compute $Q_{S^*}$, an estimation of the outcome probability $P_{S^*}$ (Eq.~\eqref{eq:P_S}), such that:
\begin{equation}
    |Q_{S^*}-P_{S^*}|\leq \exp{\left(- c N \log N- \Omega(N)\right)},
\end{equation}
for some $c\geq 1$\footnote{Throughout this work, we use standard mathematical notations describing asymptotic behavior of functions: we say $f(n) = O(g(n))$ if, for any positive constant $c$, there exists a large enough value of $n$ such that $|f(n)| \leq c g(n)$. Also, $f(n) = o(g(n))$ if, for any positive constant $c$, $|f(n)| < k g(n)$ for large enough $n$. Moreover, if $f(n) = O(g(n))$, then $g(n) = \Omega(f(n))$.}.
\end{prob}
The fact that the unitary is chosen uniformly at random from the Haar measure  is instrumental in the demonstration that this problem is $\#$P-hard on average, assuming certain conjectures~\cite{deshpande2022GBS, aaronson2011computational}. The condition on the average photon number $\bar{N}= o(\sqrt{M})$ is imposed so as to ensure that outcome events with \emph{collisions}, i.e., which can have more than one photon in a single mode, are unlikely. Moreover, the condition $K<N$ is a technical requirement to ensure that the matrices whose Hafnian needs to be computed have full rank. We refer the reader to Ref.~\cite{deshpande2022GBS} for further details on the evidence supporting the conjectures required for the hardness proof depending on the value of $c$ (for $c=6$ the evidence given is stronger than for $c=1$). We chose to include $c$ explicitly in the formulation of the problem since stronger complexity-theoretic results may be proved in the future.

\subsection{Complexity of GPNL(1) amplitudes}

Consider the problem of evaluating GPNL(1) amplitudes of the form
\begin{equation}\label{eq:At}
A_t=\bra{\Psi_{in}} \hat{U}^{\dagger} e^{i \hnl t} \hat{U} \ket{\Psi_{in}}, 
\end{equation}
as shown in Fig.~\ref{fig:equivNLBosonSamp}. We choose $\hnl$ as 
\begin{equation}\label{eq_non_linear_Ham}
    \hnl= \sum_{k=1}^{L} \eta_i \hat{n}_i^2+ \sum_{k=1}^M \mu_i \hat{n}_i, 
\end{equation}
i.e., a Hamiltonian with $L\leq M$ self-Kerr interaction terms added to a linear part in the number operators  $\hat{n}_i$. We assume, without loss of generality, that the non-linearities are placed in the first $L$ modes.  Moreover, we assume that $\eta_i$ and $\mu_i$ are real values, bounded by $|\eta_i|, |\mu_i|\leq \text{poly}(M)$. We remark that we chose to include some linear terms in the number operators in $\hnl$ (instead of absorbing it in the Gaussian transformation or the Gaussian measurement) for technical reasons related to the proof of Lemma~\ref{lemma:non_degenerate} below.

\begin{figure*}
\centering
\includegraphics[width=1.5\columnwidth]{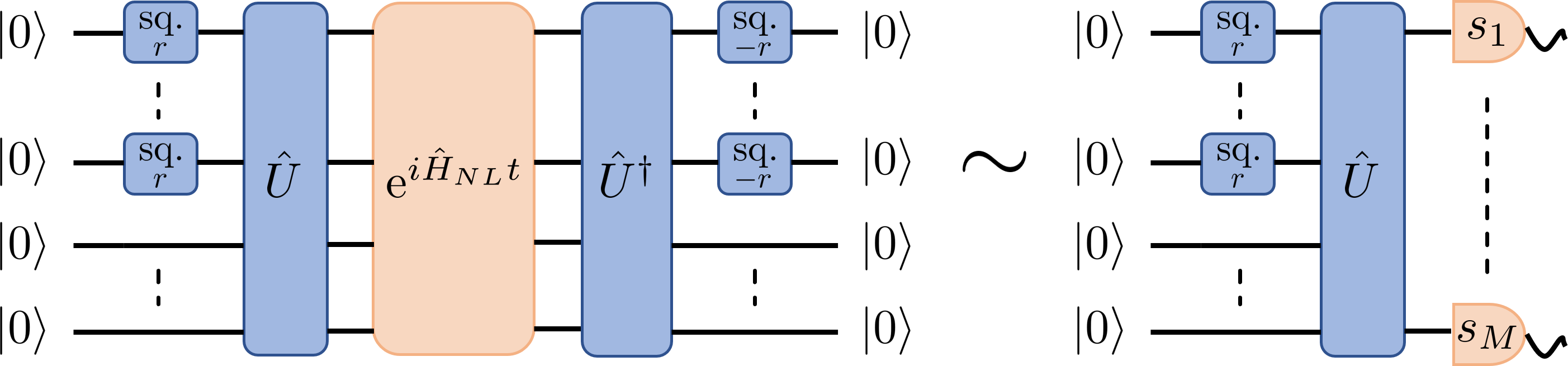}
\caption{\label{fig:equivNLBosonSamp} By computing transition amplitudes of the GPNL(1) process depicted on the left, for different evolution times $t$, we can approximate outcome probabilities of the Gaussian Boson sampler depicted on the right. The blue (darker) components represent Gaussian resources (transformations or measurements), while the orange (lighter) components represents non-Gaussian ones.}
\end{figure*}

The precise amplitude estimation problem which we connect to the problem of estimating GBS probabilities is the following.
\begin{prob}[Approx.~amplitudes~of~GPNL(1)]\label{prob:GNLG}
For a non-linear Hamiltonian $\hnl$ with $L$ non-linearities such that $L<K<M$, the problem is to estimate $A_t$ by computing a quantity $\tilde{A_t}$ such that:
\begin{equation}
    |\tilde{A}_t-A_t|\leq \exp{\left(- c L \log L- \Omega(L)\right)},
\end{equation}
for some $c\geq 1$, where $t$ is some constant independent of  $K, L $ or $M$. 
\end{prob}
\noindent We now prove the following result. 
\begin{theo}[Complexity of GPNL(1) amplitudes] \label{thm:complexity_amplitudes}
An efficient ($\text{poly}(L)$-time) classical algorithm to approximate GPNL(1) amplitudes from Problem~\ref{prob:GNLG} would imply an efficient ($\text{poly}(N)$-time) classical algorithm to approximate GBS outcome probabilities (Problem~\ref{prob:GBS}), for any $N$ such that $\bar{N}/2\leq N \leq 3 \bar{N}/2$, where $\bar{N}= K \sinh^2(r)$ is the average photon number. 
\end{theo}
\noindent Note that the constraint $\bar{N}/2\leq N \leq 3\bar{N}/2$ is very mild since for large $K$ and fixed squeezing $r$, the photon number distribution asymptotically converges to a Gaussian distribution with mean $\bar{N}$ and standard deviation $\sigma= \sqrt{\bar{N}}\cosh{r}$. Thus, this constraint is fulfilled with a probability tending to $1$ asymptotically.
\begin{proof}[Proof of Theorem~\ref{thm:complexity_amplitudes}]
The main idea to connect Problem~\ref{prob:GBS} and Problem~\ref{prob:GNLG} is a technique for eigenvalue estimation using the amplitudes $A_t$~\cite{somma2002, somma2019}. Let us write the spectral decomposition of $\hnl$ as $\hnl= \sum_j \lambda_j \Pi_j$,
where $\Pi_j$ is the projector on the eigenspace with eigenvalue $\lambda_j$. Furthermore, we choose $\mu_i$ and $\eta_i$ in Eq.~\eqref{eq_non_linear_Ham} to be positive integers so that the eigenvalues $\lambda_j$ are also positive integers. We also define $\Theta$ as the set of eigenvalues of $\hnl$. 
Then we can write the amplitudes
\begin{equation} \label{eq:Fourier_relation}
    A_t = \sum_{j=0}^{\infty} p_j e^{i j t},
\end{equation}
where the $p_j$'s are the probabilities of observing energy eigenvalue $j$ in the state $\hat{U} \ket{\Psi_{in}}$, i.e., $p_j = \bra{\Psi_{in}}\hat{U}^{\dagger}\Pi_j \hat{U} \ket{\Psi_{in}}$, if $j \in \Theta$, and $p_j = 0$ otherwise.
It is important to note that, since $\hnl$ is diagonal in the Fock basis, the probabilities $p_j$ are either $0$ or they are given by sums of outcome probabilities of a Gaussian boson sampler with input state $\ket{\Psi_{in}}$ and interferometer $\hat{U}$. Precisely, if we define the energy of a product of Fock states $\ket{S}=\ket{s_1...s_M}$ as
\begin{equation}
    E_{NL}(S)=\braket{S|\hnl|S}=\sum_{k=1}^{L} \eta_i s_i^2+ \sum_{k=1}^M \mu_i s_i,
\end{equation}
and if $j\in \Theta$, we can write 
\begin{equation}\label{eq:p_j}
 p_j= \sum_{S:~ E_{NL}(S)=j}|\braket{S|\hat{U}|\Psi_{in}}|^2.
\end{equation}
Eq.~\eqref{eq:Fourier_relation} suggests that we can extract the values of $p_j$ by an (inverse) Fourier transform of $A_t$ evaluated at different times. However, to prove Theorem~\ref{thm:complexity_amplitudes}, we need two important requirements: $(i)$~we need to extract a \emph{single outcome probability} of a Gaussian boson sampler $P_{S^*}$ of some collision free event $S^*$; $(ii)$~we need to recover this probability up to an exponentially small additive error by approximating $A_t$ only a polynomial number of times.

To address the first point, let us introduce the following Lemma (see Appendix~\ref{app:lem1} for the proof).
\begin{lem} \label{lemma:non_degenerate}
Consider a \emph{collision free} outcome $\ket{S^*}$ containing $N$ photons, where each mode contains at most one photon. Without loss of generality, we assume the photons are contained in the first $N$ modes, i.e., $\ket{S^*}= \ket{1}^{\otimes N}\ket{0}^{\otimes (M-N)}$.
The Hamiltonian 
    \begin{equation}\label{eq:H_NL_nondeg}
    \hat{H}_{NL} = N^2 \sum_{j=1}^N \hat{n}_j + \sum_{j=1}^N \hat{n}_j^2 + N^2(N+2) \sum_{j=N+1}^M  \hat{n}_j
\end{equation}
has $\ket{S^*}$ as an eigenstate with a non-degenerate eigenvalue. 
\end{lem}
\noindent Considering $\hat{H}_{NL}$ from Eq.~\eqref{eq:H_NL_nondeg}, the energy associated to $\ket{S^*}$ is $j^*= N(N^2+1)$. Since the eigenvalue is non-degenerate, Eq.~\eqref{eq:p_j} implies $p_{j^*}= P_{S^*}= |\braket{S^*|\hat{U}|\Psi_{in}}|^2$,
which is a single GBS outcome probability. The next step is to show how to recover this probability up to an exponentially small additive error. To do so, we use the fact that for a large enough value $j > J_{max}$, the probabilities $p_{j}$ decrease exponentially. Namely, we consider the following approximation of $P_{S^*}$:
\begin{equation} \label{eq:approx_PS}
    \begin{aligned}
        Q_{S^*}& = \frac{1}{J_{\max}}\sum_{k=0}^{J_{\max}-1} A_{\frac{2 \pi k}{J_{\max}}}e^{-i 2 \pi k j^*/J_{\max}} \\
        & = \sum_{j=0}^\infty p_j\sum_{k=0}^{J_{\max}-1}\dfrac{e^{i \frac{2 \pi k (j-j^*)}{J_{\max}}}}{J_{\max}} \\
        & = p_{j^*}+ \sum_{j=1}^{\infty}p_{j^*+j J_{\max}}.
    \end{aligned}
\end{equation}
A simple bound for the approximation error is then
\begin{align}\label{eq:approx_error}
    |P_{S^*}-Q_{S^*}| \leq \sum_{j\geq J_{\max}}p_j\equiv \eps.
\end{align}
The following Lemma provides a bound on $\eps$ (see Appendix~\ref{app:lem2} for the proof).
\begin{lem}\label{lem:truncation}
Assume that the number of photons $N$ is large enough such that $N^2(N+2) > N^*$,
with
\begin{equation}
    N^* = (4 \sinh^2(r)+2)\left(\frac{\log(2)}{2}K + c N \log(N)\right).
\end{equation}
Then, if we choose $J_{\max}= N^4(N+2)^2=O(N^6)$, we have
\begin{equation}
    \sum_{j\geq J_{\max}}p_j\leq \exp(- c N\log(N)).
\end{equation}
\end{lem}
\noindent Note that $N^2(N+2) > N^*$ is asymptotically fulfilled under the assumption on $N$ in Theorem~\ref{prob:GBS}.
Since $J_{\max}= \text{poly}(N)$, it can be seen from the first identity in Eq.~\eqref{eq:approx_PS} that $Q_{S^*}$ can be computed using only a polynomial number of \emph{exact} computations of the amplitudes $A_{\frac{2 \pi k}{J_{\max}}}$.

Finally, to prove Theorem~\ref{thm:complexity_amplitudes}, we need to consider how an exponentially small error in the estimation of these amplitudes affects the total error. It is possible to see that if we compute 
\begin{align}
    \tilde{Q}_{S^*}&= \frac{1}{J_{\max}}\sum_{k=0}^{J_{\max}-1} \tilde{A}_{\frac{2 \pi k}{J_{\max}}}e^{-i 2 \pi k j^*/J_{\max}},  \label{eq:approx2_PS}
\end{align}
using approximations of the amplitudes such that 
\begin{equation}
    |\tilde{A}_{\frac{2 \pi k}{J_{\max}}}-A_{\frac{2 \pi k}{J_{\max}}}|\leq \exp{(- c N\log(N) - \Omega(N))}, ~\forall~k,
\end{equation}
we still obtain a sufficiently small error
\begin{equation}
    |P_{S^*}-\tilde{Q}_{S^*}| \leq \exp{(- c N\log(N) - \Omega(N))},
\end{equation}
if we choose $J_{\max}$ according to Lemma~\ref{lem:truncation}. Thus, if the evaluation of $\tilde{A}_{\frac{2 \pi k}{J_{\max}}}$ up to exponentially small error could be done in polynomial time, the probabilities $\tilde{Q}_{S^*}$ could be obtained also in polynomial time and it would be possible to solve Problem~\ref{prob:GBS} efficiently. \end{proof}

\section{Continuous-variable Hadamard test}

Given the result of Theorem~\ref{thm:complexity_amplitudes}, it is natural to ask whether there is evidence that not only  amplitudes related to GPNL(1), but also probabilities are exponentially hard to compute. While we leave this as an open question, we show evidence of classical hardness of outcome probabilities of Gaussian processes with two layers of non-linearities (GPNL(2)). In this section, we show how such processes can be used to estimate the amplitudes $A_t$ and thus should be hard to simulate classically given the result from Theorem~\ref{thm:complexity_amplitudes}.

\begin{figure*}[t]
\centering
	\includegraphics[width=1.7\columnwidth]{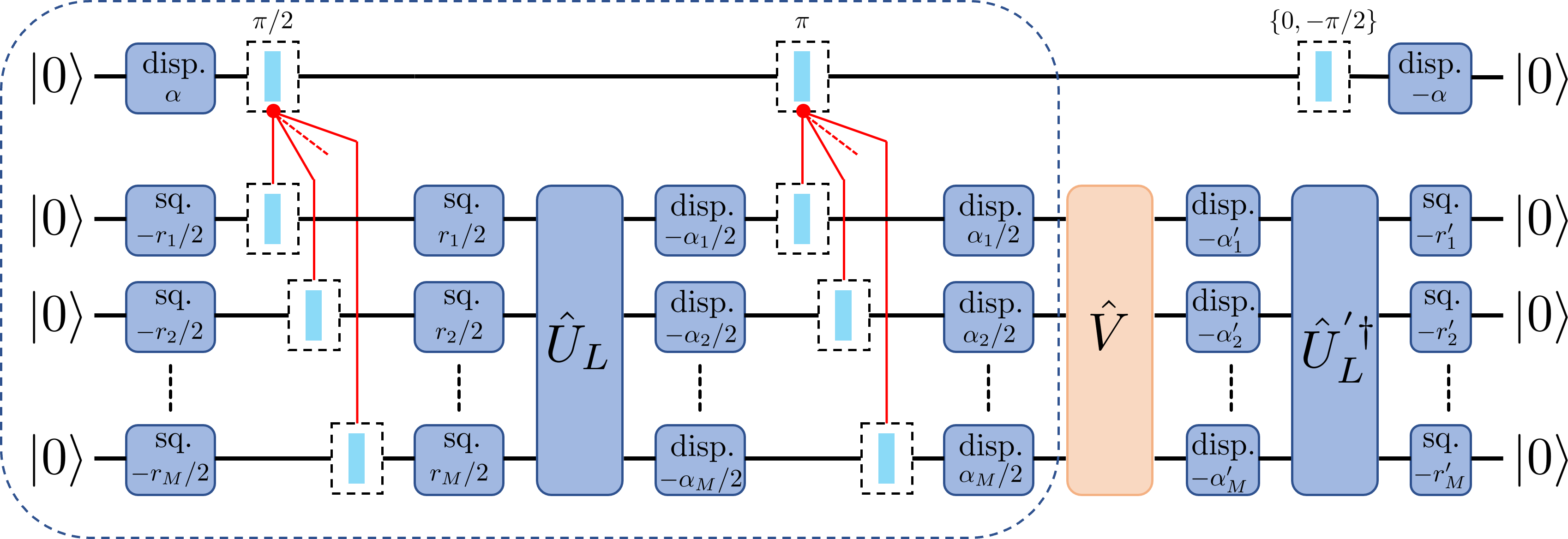}
	\caption{\label{fig:CVHadamard} Continuous-variable Hadamard test  : protocol to estimate amplitudes of the type $A=\bra{\Psi_G'} \hat{V} \ket{\Psi_G}$, where $\ket{\Psi_G}$ and $\ket{\Psi_G'}$ are $M$-mode Gaussian states and $\hat{V}$ can be any unitary that conserves the total number of photons. In the Figure, we choose the general decompositions $\ket{\Psi_G} = \left( \bigotimes_{m=1}^M \hat{D}_{\alpha_m} \right) \hat{U}_L \left( \bigotimes_{l=1}^M \hat{S}_{r_l} \right) \ket{\boldsymbol{0}}$ and $\ket{\Psi'_G} = \left( \bigotimes_{m=1}^M \hat{D}_{\alpha'_m} \right) \hat{U}'_L \left( \bigotimes_{l=1}^M \hat{S}_{r'_l} \right) \ket{\boldsymbol{0}}$. The dashed envelope corresponds to the part of the circuit that prepares the non-Gaussian superposition $\ket{\Lambda} = \ket{\phi_+} \ket{\boldsymbol{0}} + \ket{\phi_-} \ket{\Psi_G}$ of Eq.~\eqref{eq:LambdaApp}, where $\ket{\phi_+}$ and $\ket{\phi_-}$ are the subnormalized even and odd cat states defined in Eq.~\eqref{eq:phi0phi1App}.}
\end{figure*}

To do so, we first develop a general procedure to estimate amplitudes of the type $A=\bra{\Psi_G'} \hat{V}\ket{\Psi_G}$, where $\ket{\Psi_G}$ and $\ket{\Psi_G'}$ are $M$-mode Gaussian states and $\hat{V}$ is any unitary that conserves the total photon number, such as the time-evolution under the Hamiltonians from Eqs.~\eqref{eq:selfK}  and~\eqref{eq:crossK}. The procedure is inspired by the Hadamard test for qubits~\cite{Cleve1998}. It requires an ancillary input mode in a coherent state $\ket{\alpha}$, as well as the ability to perform cross-Kerr interactions between the ancillary mode and each of the $M$-modes via the unitary $\hat{U}_{CP}(\varphi) = e^{- i \varphi \hat{n}_0 \sum_{k=1}^{M} \hat{n}_k}$.
Here, the number operator $\hat{n}_0$ acts on the ancilla mode while the number operators $\hat{n}_k$ act on the remaining $M$ system modes. The operator $\hat{U}_{CP}(\varphi)$ can be seen as a controlled-phase gate, as it applies single-mode phase shifts controlled on the number of photons in the ancillary mode.
The protocol is represented on Fig.~\ref{fig:CVHadamard} and it involves three main steps which are outlined in what follows (see Appendix~\ref{app:CVHadamard} for details):

$(i)$~Starting from the $(M+1)$-mode Gaussian state $\ket{\alpha}\ket{\boldsymbol{0}}$, where  $\ket{\boldsymbol{0}}$ is the $M$-mode vaccum state, prepare the superposition 
\begin{equation}
\ket{\Lambda} = \ket{\phi_+} \ket{\boldsymbol{0}} + \ket{\phi_-} \ket{\Psi_G},
\end{equation}
where $\ket{\phi_{\pm}} = (\ket{\alpha} \pm \ket{-\alpha})/2$ are subnormalized even and odd cat states~\cite{dodonov1974even};

$(ii)$~Apply the unitary $\hat{V}$, acting non-trivially on the $M$ system modes, to obtain the state:
\begin{equation}
 (\mathbb{I}\otimes\hat{V})\ket{\Lambda}= \ket{\phi_+}  \ket{\boldsymbol{0}} + \ket{\phi_-} \hat{V} \ket{\Psi_G},   
\end{equation} 
where we assumed $\hat{V}\ket{\boldsymbol{0}}=\ket{\boldsymbol{0}}$;

$(iii)$~Perform a Gaussian measurement to estimate the probability
\begin{equation}\label{eq:final_probability}
    \begin{aligned}
        |\bra{\alpha} \bra{\Psi_G'} \mathbb{I}\otimes\hat{V}\ket{\Lambda}|^2 & = c_1 | \bk{\Psi_G'}{\boldsymbol{0}} |^2 + c_2 | \bra{\Psi_G'} \hat{V} \ket{\Psi_G} |^2 \\
        & \quad + c_3 \, \mathrm{Re}[ \bk{\Psi_G'}{\boldsymbol{0}} \bra{\Psi_G} \hat{V} \ket{\Psi_G'}],
    \end{aligned}
\end{equation}
where $c_1, c_2$ and $c_3$ are functions of $\alpha$ defined above Eq.~\eqref{eq:c1c2c3} of Appendix~\ref{app:CVHadamard}.

To prepare state $\ket{\Lambda} $ in step $(i)$, we assume we can prepare $\ket{\Psi_G}= \hat{U}_G \ket{\boldsymbol{0}}$, where $\hat{U}_G$ is a Gaussian unitary.
We show in Appendix~\ref{app:CVHadamard} that, apart from Gaussian operations, step $(i)$ requires a single call to the controlled-phase operator $\hat{U}_{CP}(\varphi)$ if $\hat{U}_G$ contains no displacement, and two calls to $\hat{U}_{CP}(\varphi)$ otherwise.

Overall, this leads to a conceptually simple protocol to estimate the complex-valued  amplitude $\bra{\Psi_G'} \hat{V} \ket{\Psi_G}$. By inspecting Eq.~\eqref{eq:final_probability} we notice that the overlap $\bk{\Psi_G'}{\boldsymbol{0}}$ can be computed efficiently analytically. Moreover, the probability $| \bra{\Psi_G'} \hat{V} \ket{\Psi_G} |^2$ can be estimated independently. By plugging in the values of these quantities, we obtain the value of $\mathrm{Re}[ \bra{\Psi_G'} \hat{V} \ket{\Psi_G}]$ from the measured probability in step $(iii)$.
Note that $| \bra{\Psi_G'} \hat{V} \ket{\Psi_G} |^2$ can be seen as the overlap between two states $\hat{V} \ket{\Psi_G}$ and $\ket{\Psi'_{G}}$, so it can be estimated by preparing $\hat{V}\ket{\Psi_G}$ and estimating the probability of observing the state $\ket{\Psi_{G}}$ via a Gaussian measurement. Other protocols to estimate state overlaps involve the preparing the two states,  applying a beam-splitter unitary and performing Fock state measurements \cite{daley2012measuring}, or coupling the two systems to an ancillary two-level system in a SWAP test \cite{Filip2002}, but these go beyond our framework which only assumes Gaussian measurements.

\section{Complexity of GPNL(2) probabilities}

The precise way we can use GPNL(2) to estimate amplitudes $A_t$ from Eq.~\eqref{eq:At} is represented in Fig.~\ref{fig:CVHadamardMain}. This circuit uses an ancillary mode, in order to estimate the amplitudes $A_t$ defined in Eq.~\eqref{eq:At}.  Moreover, from Theorem~\ref{thm:complexity_amplitudes} we know how to estimate GBS outcome probabilities from a polynomial number of amplitudes $A_t$. Combining the two arguments, we obtain the following theorem which shows hardness of approximating GPNL(2) probabilities up to an exponentially small error. 
\begin{theo}[Complexity of GPNL(2) probabilities] \label{thm:complexity_GPNL2}
Assume the existence of an efficient ($\text{poly}(M)$-time) classical algorithm to approximate outcome probabilities of GPNL(2) of the form
   \begin{equation} \label{eq:thm:complexity_GPNL2}
    |\bra{\Psi_G}\prod_{j=1}^2 \exp(i H_{NL}^{(j)}\Delta t) U_G^{(j)}
    \ket{0}^{\otimes M}|^2
\end{equation} 
up to an additive error $O(\exp(- c M \log(M)))$. This would imply the existence of  an efficient ($\text{poly}(N)$-time) classical algorithm to approximate the GBS outcome probabilities (Problem~\ref{prob:GBS}). 
\end{theo}

\begin{proof}[Proof of Theorem~2]
Suppose there is an efficient ($\text{poly}(M)$-time) classical algorithm to approximate probabilities of the form
\begin{equation}
    |\bra{\Psi_G}\prod_{j=1}^2 \exp(i H_{NL}^{(j)}\Delta t) U_G^{(j)}
    \ket{0}^{\otimes M}|^2,
\end{equation} 
up to an additive error $\exp(- c M \log(M))$. This would imply that the outcome probabilities of the CV Hadamard test presented in the main text could be efficiently  approximated up to this error, for $\hat{V}= \exp(- i t \hat{H}_{NL})$, since the overall circuit can be seen as a Gaussian process with two layers of non-linearities -- one from the controlled-phase gate $\hat{U}_{CP}(\pi/2)$ and one from the time-evolution under $\hat{H}_{NL}$. Choosing $\ket{\Psi_G}=\ket{\Psi'_G}= \ket{\Psi_{out}}=\hat{U}\ket{\psi_{\mathrm{SMS}}(r)}^{\otimes K }\ket{0}^{\otimes (M-K)}$ to be the output state of the Gaussian boson sampler, we have from Eq.~\eqref{eq:final_probability_appCVHad} that
\begin{figure}
\centering
\includegraphics[width=0.9\columnwidth]{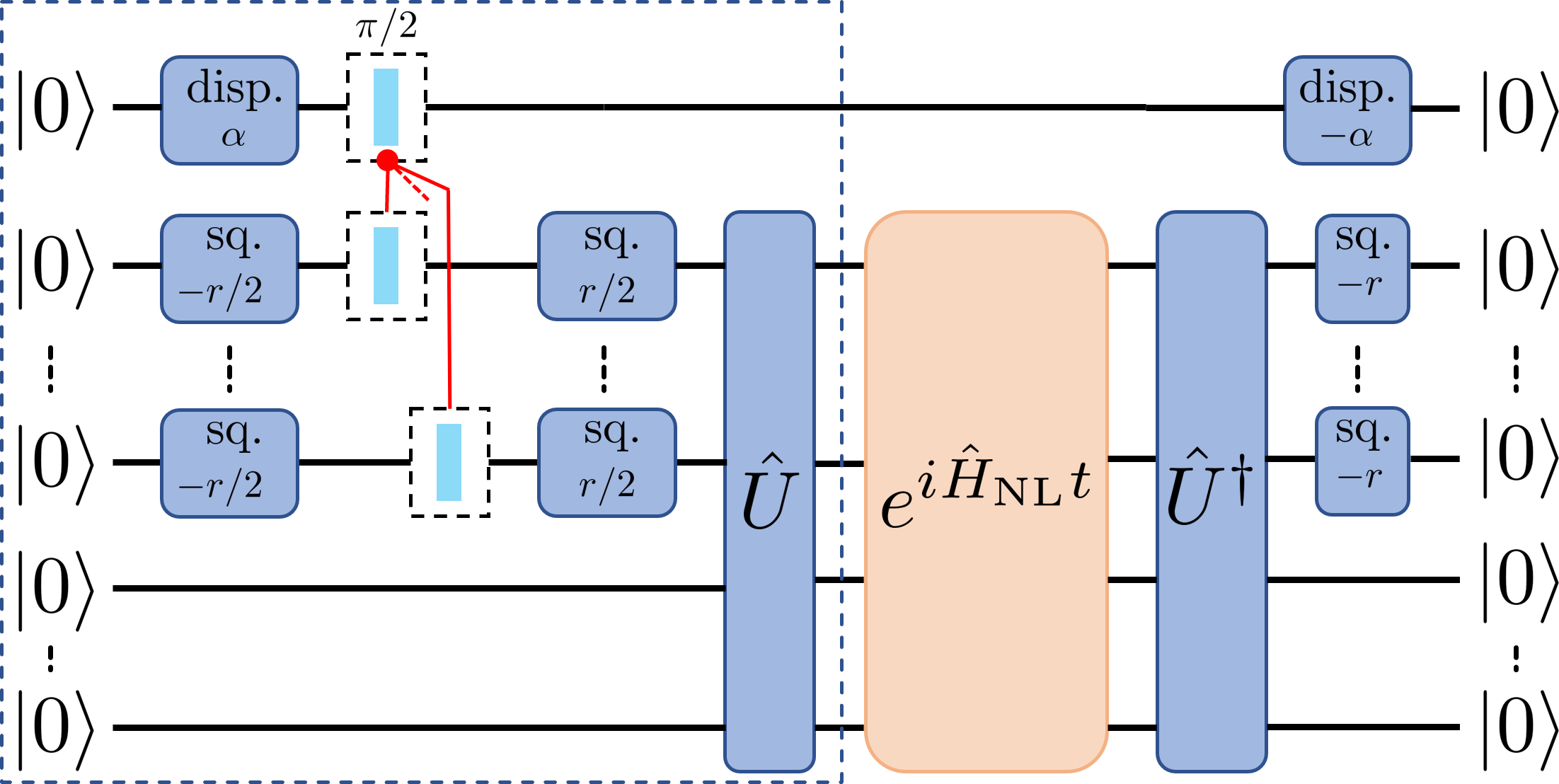}
\caption{\label{fig:CVHadamardMain} An instance of the CV Hadamard test, used to estimate the amplitude $A_t$. The red lines represent the controlled operation $\hat{U}_{CP}(\pi/2)$ and the dashed envelope corresponds to the part of the circuit that prepares the state $\ket{\Lambda}$.}
\end{figure}
\begin{equation}\label{eq:final_probability_app}
    \begin{aligned}
        &|\bra{\alpha} \bra{\Psi_G} \mathbb{I}\otimes\hat{V}\ket{\Lambda}|^2  =\\
        & = c_1 | \bk{\Psi_{out}}{\boldsymbol{0}} |^2 + c_2 | A_t |^2 + c_3 \, \mathrm{Re}[ \bk{\Psi_{out}}{\boldsymbol{0}} A_t] \\
        & = c_1 (\cosh(r))^{-K} + c_2 |A_t|^2 + c_3 (\cosh(r))^{-K/2} \, \mathrm{Re}[A_t],
    \end{aligned}
\end{equation}
where $c_1$, $c_2$ and $c_3$ are functions of $\alpha$ (see Eq.~\eqref{eq:final_probability_appCVHad}). Denoting $P^{(2)}_t = |\bra{\alpha} \bra{\Psi_G} \mathbb{I}\otimes\hat{V}\ket{\Lambda}|^2$ and $P^{(1)}_t = |A_t|^2$, we can write
\begin{align}
    \mathrm{Re}[A_t] =& \frac{(\cosh(r))^{K/2}}{c_3} P^{(2)}_t - \frac{c_2 (\cosh(r))^{K/2}}{c_3}  P^{(1)}_t \nonumber\\ &- \frac{c_1}{c_3} (\cosh(r))^{-K/2}.
\end{align}
From the hypothesis of the Theorem, we can compute $\tilde{P}^{(1)}_t$ such that
\begin{equation}
    |\tilde{P}^{(1)}_t-P^{(1)}_t| \leq \exp{\left(-c M \log M \right)},
\end{equation}
and similarly we can compute $\tilde{P}^{(2)}_t$ such that
\begin{equation}
    |\tilde{P}^{(2)}_t-P^{(2)}_t| \leq \exp{\left(- cM \log M \right)}.
\end{equation}
We can therefore also compute
\begin{align}
    \tilde{R}_t =& \frac{(\cosh(r))^{K/2}}{c_3} \tilde{P}^{(2)}_t - \frac{c_2 (\cosh(r))^{K/2}}{c_3} \tilde{P}^{(1)}_t \\& - \frac{c_1}{c_3} (\cosh(r))^{-K/2},
\end{align}
and we have
\begin{equation}
    \begin{aligned}
        &|\tilde{R}_t - \mathrm{Re}[A_t]|
         \leq  \\ &\leq \frac{(\cosh(r))^{K/2}}{|c_3|} \left( |\tilde{P}^{(2)}_t - P^{(2)}_t| + |c_2| |\tilde{P}^{(1)}_t - P^{(1)}_t| \right) \\
        &\leq (1+|c_2|) \frac{(\cosh(r))^{K/2}}{|c_3|} \exp{\left(- cM \log M \right)} \\
         &\leq \frac{1+|c_2|}{|c_3|} e^{K r} \exp{\left(- cM \log M \right)} \\
         &\leq \exp(- N\log(N)) ,
    \end{aligned}
\end{equation}
where in the last step we use the fact that $K\leq M $ and $N<3/2 \bar{N}= o(\sqrt{M}) $ from the collision free condition.
\end{proof}

\subsection{Exact sampling hardness}

While our work focuses on understanding the complexity of approximating outcome probabilities, our results can readily be coupled with standard techniques in the literature (see for instance~\cite{aaronson2011computational,levonBS_Gaussian_meas,chabaud2017photonadded}) to show that it should be hard to sample from Gaussian processes with two layers of non-linearities. Here, we outline a sketch of the argument that can be used to convert Theorem~\ref{thm:complexity_GPNL2} into an exact sampling hardness result. We consider the problem of sampling from the quantum circuit realizing the CV Hadamard test shown on Fig.~\ref{fig:CVHadamardMain}, where the measurements are now discretized versions of Gaussian measurements, similarly to what was considered in Refs.~\cite{levonBS_Gaussian_meas, chabaud2017photonadded}. Namely, this is done by discretizing the phase space of every output mode of the circuit into boxes of particular sizes. In this setting, the probability of observing the outcome associated with vaccuum will be approximately given by Eq.~\eqref{eq:final_probability_app}, with an error associated to the discretization of the Gaussian measurement. This error can be made exponentially small using a suitable discretization~\cite{levonBS_Gaussian_meas,chabaud2017photonadded}. Hence, the probability of observing the outcome associated to vaccuum is still \#P-hard to compute according to the result of Theorem~\ref{thm:complexity_GPNL2}. Following standard arguments from~\cite{aaronson2011computational}, which make use of Stockmeyer's algorithm, this implies that no efficient classical algorithm exists to exactly sample from outcomes of the circuit represented on Fig.~\ref{fig:CVHadamardMain} with discretized Gaussian measurements, unless the polynomial hierarchy collapses.

The argument sketched above leads to an \emph{exact sampling} hardness proof. However, since no quantum device is perfect, any quantum computational advantage proof should consider the hardness of \emph{approximately sampling} from the quantum circuit, i.e., up to some fixed total variation distance. Proving approximate sampling hardness is considerably more challenging and usually requires additional mathematical conjectures~\cite{aaronson2011computational,deshpande2022GBS}.

A limitation of the technique we use to show Theorem~\ref{thm:complexity_GPNL2} is that we can only connect a single outcome probability  from the circuit of Fig.~\ref{fig:CVHadamardMain} (the vaccuum probability) to a known hard problem (approximation of GBS probabilities). It is unclear from our techniques whether it is possible to argue that the probabilities of other possible outcomes of Gaussian measurements can also be \#P-hard to compute, which is a necessary step towards showing approximate sampling hardness. Hence, we leave the question of understanding the complexity of approximately sampling from Gaussian processes with a one or two layers of non-linearities as subject of future research.

\section{Discussion}

The demonstration that linear optical systems can be difficult to simulate classically has sparked a huge theoretical and experimental effort in the last decade. Here, we show that the introduction of a limited number of layers of non-linear interactions of the Kerr type can greatly affect the complexity of simulating quantum optical systems. Our main result is that a \emph{single layer} of non-linearities is enough to turn a classically simulable system (a Gaussian quantum process), into one where computing transition amplitudes is at least as hard as computing probabilities of a GBS experiment. Our result, combined with the results from Ref.~\cite{grier2022bipartiteGBS}, implies that a polynomial number of GPNL(1) amplitudes can be used to estimate the absolute value of the permanent of an arbitrary matrix. We leave as an open question whether it is possible to also prove the hardness of computing transition \emph{probabilities} of GPNL(1). This might require ideas going beyond the main ingredient of the proof of Theorem~\ref{thm:complexity_amplitudes}, which uses the time-series approach for quantum phase estimation~\cite{somma2002, somma2019} to establish the connection between amplitudes of GPNL(1) and GBS outcome probabilities.

In order to prove the hardness of computing outcome probabilities, we introduced an extra layer of non-linearities of the cross-Kerr type and showed that outcome probabilities of GPNL(2) can be used to estimate amplitudes of GPNL(1). More generally, we developed a version of the Hadamard test that is adapted for CV systems. Given the importance of the Hadamard test in quantum computing and, in particular, as a possible way to solve the problem of Hamiltonian eigenvalue estimation with less coherence requirements~\cite{somma2002, somma2019}, we believe this result to be of independent interest. Implementations would in principle be possible in CV quantum systems with tunable cross-Kerr interactions, such as cavity QED~\cite{crossKerrcQED2018}.  

While our work focuses on the complexity of \emph{strong simulation} of GPNL($k$), we leave the question of understanding the complexity of \emph{approximately sampling} from random ensembles of these quantum processes open. Moreover, it would be interesting to find potential applications of such quantum circuits. As boson sampling devices have found applications to molecular vibrational spectra problems \cite{huh2015boson}, the addition of a limited number of layers of non-linearities in boson sampling set-ups may potentially be used to probe the spectra of more complex, interacting bosonic Hamiltonians, potentially outperforming competing classical methods \cite{oh2024quantum}. Moreover, the addition of non-linearities may also increase expressivity of bosonic quantum circuits in the context of quantum machine learning, in a similar fashion as the addition of layers of adaptivity \cite{chabaud2021quantum}.
Pursuing such questions should not only lead to a better understanding of classical simulability of bosonic systems, but also it may inspire new approaches for quantum computational advantage and new algorithmic applications of near-term bosonic quantum devices with a limited number of interaction terms.

\begin{acknowledgments}
We thank Daniel Brod for useful comments on the manuscript.
L.N. acknowledges useful discussions with Ulysse Chabaud and Giulia Ferrini. 
M.G.J. acknowledges support by the Belgian Fund for Scientific Research (F.R.S.--FNRS) and by the Carlsberg Foundation under Grant CF19-0313.
L.N. acknowledges support by the Belgian Fund for Scientific Research (F.R.S.--FNRS), FCT-Fundação para a Ciência e a Tecnologia (Portugal) via the Project No. CEECINST/00062/2018 and from the European Union's Horizon 2020 research and innovation program through the FET project PHOQUSING (“PHOtonic QUantum SamplING machine” - Grant Agreement No. 899544). 
\end{acknowledgments}

\bibliography{compNonLinBoson}

\begin{thebibliography}{37}%
\makeatletter
\providecommand \@ifxundefined [1]{%
 \@ifx{#1\undefined}
}%
\providecommand \@ifnum [1]{%
 \ifnum #1\expandafter \@firstoftwo
 \else \expandafter \@secondoftwo
 \fi
}%
\providecommand \@ifx [1]{%
 \ifx #1\expandafter \@firstoftwo
 \else \expandafter \@secondoftwo
 \fi
}%
\providecommand \natexlab [1]{#1}%
\providecommand \enquote  [1]{``#1''}%
\providecommand \bibnamefont  [1]{#1}%
\providecommand \bibfnamefont [1]{#1}%
\providecommand \citenamefont [1]{#1}%
\providecommand \href@noop [0]{\@secondoftwo}%
\providecommand \href [0]{\begingroup \@sanitize@url \@href}%
\providecommand \@href[1]{\@@startlink{#1}\@@href}%
\providecommand \@@href[1]{\endgroup#1\@@endlink}%
\providecommand \@sanitize@url [0]{\catcode `\\12\catcode `\$12\catcode
  `\&12\catcode `\#12\catcode `\^12\catcode `\_12\catcode `\%12\relax}%
\providecommand \@@startlink[1]{}%
\providecommand \@@endlink[0]{}%
\providecommand \url  [0]{\begingroup\@sanitize@url \@url }%
\providecommand \@url [1]{\endgroup\@href {#1}{\urlprefix }}%
\providecommand \urlprefix  [0]{URL }%
\providecommand \Eprint [0]{\href }%
\providecommand \doibase [0]{http://dx.doi.org/}%
\providecommand \selectlanguage [0]{\@gobble}%
\providecommand \bibinfo  [0]{\@secondoftwo}%
\providecommand \bibfield  [0]{\@secondoftwo}%
\providecommand \translation [1]{[#1]}%
\providecommand \BibitemOpen [0]{}%
\providecommand \bibitemStop [0]{}%
\providecommand \bibitemNoStop [0]{.\EOS\space}%
\providecommand \EOS [0]{\spacefactor3000\relax}%
\providecommand \BibitemShut  [1]{\csname bibitem#1\endcsname}%
\let\auto@bib@innerbib\@empty
\bibitem [{\citenamefont {Weedbrook}\ \emph {et~al.}(2012)\citenamefont
  {Weedbrook}, \citenamefont {Pirandola}, \citenamefont {Garc\'{\i}a-Patr\'on},
  \citenamefont {Cerf}, \citenamefont {Ralph}, \citenamefont {Shapiro},\ and\
  \citenamefont {Lloyd}}]{Weedbrook2012}%
  \BibitemOpen
  \bibfield  {author} {\bibinfo {author} {\bibfnamefont {C.}~\bibnamefont
  {Weedbrook}}, \bibinfo {author} {\bibfnamefont {S.}~\bibnamefont
  {Pirandola}}, \bibinfo {author} {\bibfnamefont {R.}~\bibnamefont
  {Garc\'{\i}a-Patr\'on}}, \bibinfo {author} {\bibfnamefont {N.~J.}\
  \bibnamefont {Cerf}}, \bibinfo {author} {\bibfnamefont {T.~C.}\ \bibnamefont
  {Ralph}}, \bibinfo {author} {\bibfnamefont {J.~H.}\ \bibnamefont {Shapiro}},
  \ and\ \bibinfo {author} {\bibfnamefont {S.}~\bibnamefont {Lloyd}},\ }\href
  {\doibase 10.1103/RevModPhys.84.621} {\bibfield  {journal} {\bibinfo
  {journal} {Review of Modern Physics}\ }\textbf {\bibinfo {volume} {84}},\
  \bibinfo {pages} {621} (\bibinfo {year} {2012})}\BibitemShut {NoStop}%
\bibitem [{\citenamefont {Bartlett}\ \emph {et~al.}(2002)\citenamefont
  {Bartlett}, \citenamefont {Sanders}, \citenamefont {Braunstein},\ and\
  \citenamefont {Nemoto}}]{sim_Gaussian}%
  \BibitemOpen
  \bibfield  {author} {\bibinfo {author} {\bibfnamefont {S.~D.}\ \bibnamefont
  {Bartlett}}, \bibinfo {author} {\bibfnamefont {B.~C.}\ \bibnamefont
  {Sanders}}, \bibinfo {author} {\bibfnamefont {S.~L.}\ \bibnamefont
  {Braunstein}}, \ and\ \bibinfo {author} {\bibfnamefont {K.}~\bibnamefont
  {Nemoto}},\ }\href {\doibase 10.1103/PhysRevLett.88.097904} {\bibfield
  {journal} {\bibinfo  {journal} {Phys. Rev. Lett.}\ }\textbf {\bibinfo
  {volume} {88}},\ \bibinfo {pages} {097904} (\bibinfo {year}
  {2002})}\BibitemShut {NoStop}%
\bibitem [{\citenamefont {Hangleiter}\ and\ \citenamefont
  {Eisert}(2023)}]{hangleiter2023computational}%
  \BibitemOpen
  \bibfield  {author} {\bibinfo {author} {\bibfnamefont {D.}~\bibnamefont
  {Hangleiter}}\ and\ \bibinfo {author} {\bibfnamefont {J.}~\bibnamefont
  {Eisert}},\ }\href@noop {} {\bibfield  {journal} {\bibinfo  {journal}
  {Reviews of Modern Physics}\ }\textbf {\bibinfo {volume} {95}},\ \bibinfo
  {pages} {035001} (\bibinfo {year} {2023})}\BibitemShut {NoStop}%
\bibitem [{\citenamefont {Hamilton}\ \emph {et~al.}(2017)\citenamefont
  {Hamilton}, \citenamefont {Kruse}, \citenamefont {Sansoni}, \citenamefont
  {Barkhofen}, \citenamefont {Silberhorn},\ and\ \citenamefont
  {Jex}}]{Hamilton2017}%
  \BibitemOpen
  \bibfield  {author} {\bibinfo {author} {\bibfnamefont {C.~S.}\ \bibnamefont
  {Hamilton}}, \bibinfo {author} {\bibfnamefont {R.}~\bibnamefont {Kruse}},
  \bibinfo {author} {\bibfnamefont {L.}~\bibnamefont {Sansoni}}, \bibinfo
  {author} {\bibfnamefont {S.}~\bibnamefont {Barkhofen}}, \bibinfo {author}
  {\bibfnamefont {C.}~\bibnamefont {Silberhorn}}, \ and\ \bibinfo {author}
  {\bibfnamefont {I.}~\bibnamefont {Jex}},\ }\href {\doibase
  10.1103/PhysRevLett.119.170501} {\bibfield  {journal} {\bibinfo  {journal}
  {Phys. Rev. Lett.}\ }\textbf {\bibinfo {volume} {119}},\ \bibinfo {pages}
  {170501} (\bibinfo {year} {2017})}\BibitemShut {NoStop}%
\bibitem [{\citenamefont {Aaronson}\ and\ \citenamefont
  {Arkhipov}(2011)}]{aaronson2011computational}%
  \BibitemOpen
  \bibfield  {author} {\bibinfo {author} {\bibfnamefont {S.}~\bibnamefont
  {Aaronson}}\ and\ \bibinfo {author} {\bibfnamefont {A.}~\bibnamefont
  {Arkhipov}},\ }in\ \href@noop {} {\emph {\bibinfo {booktitle} {Proceedings of
  the forty-third annual ACM symposium on Theory of computing}}}\ (\bibinfo
  {year} {2011})\ pp.\ \bibinfo {pages} {333--342}\BibitemShut {NoStop}%
\bibitem [{\citenamefont {Deshpande}\ \emph {et~al.}(2022)\citenamefont
  {Deshpande}, \citenamefont {Mehta}, \citenamefont {Vincent}, \citenamefont
  {Quesada}, \citenamefont {Hinsche}, \citenamefont {Ioannou}, \citenamefont
  {Madsen}, \citenamefont {Lavoie}, \citenamefont {Qi}, \citenamefont {Eisert}
  \emph {et~al.}}]{deshpande2022GBS}%
  \BibitemOpen
  \bibfield  {author} {\bibinfo {author} {\bibfnamefont {A.}~\bibnamefont
  {Deshpande}}, \bibinfo {author} {\bibfnamefont {A.}~\bibnamefont {Mehta}},
  \bibinfo {author} {\bibfnamefont {T.}~\bibnamefont {Vincent}}, \bibinfo
  {author} {\bibfnamefont {N.}~\bibnamefont {Quesada}}, \bibinfo {author}
  {\bibfnamefont {M.}~\bibnamefont {Hinsche}}, \bibinfo {author} {\bibfnamefont
  {M.}~\bibnamefont {Ioannou}}, \bibinfo {author} {\bibfnamefont
  {L.}~\bibnamefont {Madsen}}, \bibinfo {author} {\bibfnamefont
  {J.}~\bibnamefont {Lavoie}}, \bibinfo {author} {\bibfnamefont
  {H.}~\bibnamefont {Qi}}, \bibinfo {author} {\bibfnamefont {J.}~\bibnamefont
  {Eisert}},  \emph {et~al.},\ }\href@noop {} {\bibfield  {journal} {\bibinfo
  {journal} {Science advances}\ }\textbf {\bibinfo {volume} {8}},\ \bibinfo
  {pages} {eabi7894} (\bibinfo {year} {2022})}\BibitemShut {NoStop}%
\bibitem [{\citenamefont {Grier}\ \emph {et~al.}(2022)\citenamefont {Grier},
  \citenamefont {Brod}, \citenamefont {Arrazola}, \citenamefont
  {de~Andrade~Alonso},\ and\ \citenamefont {Quesada}}]{grier2022bipartiteGBS}%
  \BibitemOpen
  \bibfield  {author} {\bibinfo {author} {\bibfnamefont {D.}~\bibnamefont
  {Grier}}, \bibinfo {author} {\bibfnamefont {D.~J.}\ \bibnamefont {Brod}},
  \bibinfo {author} {\bibfnamefont {J.~M.}\ \bibnamefont {Arrazola}}, \bibinfo
  {author} {\bibfnamefont {M.~B.}\ \bibnamefont {de~Andrade~Alonso}}, \ and\
  \bibinfo {author} {\bibfnamefont {N.}~\bibnamefont {Quesada}},\ }\href@noop
  {} {\bibfield  {journal} {\bibinfo  {journal} {Quantum}\ }\textbf {\bibinfo
  {volume} {6}},\ \bibinfo {pages} {863} (\bibinfo {year} {2022})}\BibitemShut
  {NoStop}%
\bibitem [{\citenamefont {Zhong}\ \emph {et~al.}(2020)\citenamefont {Zhong},
  \citenamefont {Wang}, \citenamefont {Deng}, \citenamefont {Chen},
  \citenamefont {Peng}, \citenamefont {Luo}, \citenamefont {Qin}, \citenamefont
  {Wu}, \citenamefont {Ding}, \citenamefont {Hu} \emph
  {et~al.}}]{zhong2020GBS}%
  \BibitemOpen
  \bibfield  {author} {\bibinfo {author} {\bibfnamefont {H.-S.}\ \bibnamefont
  {Zhong}}, \bibinfo {author} {\bibfnamefont {H.}~\bibnamefont {Wang}},
  \bibinfo {author} {\bibfnamefont {Y.-H.}\ \bibnamefont {Deng}}, \bibinfo
  {author} {\bibfnamefont {M.-C.}\ \bibnamefont {Chen}}, \bibinfo {author}
  {\bibfnamefont {L.-C.}\ \bibnamefont {Peng}}, \bibinfo {author}
  {\bibfnamefont {Y.-H.}\ \bibnamefont {Luo}}, \bibinfo {author} {\bibfnamefont
  {J.}~\bibnamefont {Qin}}, \bibinfo {author} {\bibfnamefont {D.}~\bibnamefont
  {Wu}}, \bibinfo {author} {\bibfnamefont {X.}~\bibnamefont {Ding}}, \bibinfo
  {author} {\bibfnamefont {Y.}~\bibnamefont {Hu}},  \emph {et~al.},\
  }\href@noop {} {\bibfield  {journal} {\bibinfo  {journal} {Science}\ }\textbf
  {\bibinfo {volume} {370}},\ \bibinfo {pages} {1460} (\bibinfo {year}
  {2020})}\BibitemShut {NoStop}%
\bibitem [{\citenamefont {Zhong}\ \emph {et~al.}(2021)\citenamefont {Zhong},
  \citenamefont {Deng}, \citenamefont {Qin}, \citenamefont {Wang},
  \citenamefont {Chen}, \citenamefont {Peng}, \citenamefont {Luo},
  \citenamefont {Wu}, \citenamefont {Gong}, \citenamefont {Su}, \citenamefont
  {Hu}, \citenamefont {Hu}, \citenamefont {Yang}, \citenamefont {Zhang},
  \citenamefont {Li}, \citenamefont {Li}, \citenamefont {Jiang}, \citenamefont
  {Gan}, \citenamefont {Yang}, \citenamefont {You}, \citenamefont {Wang},
  \citenamefont {Li}, \citenamefont {Liu}, \citenamefont {Renema},
  \citenamefont {Lu},\ and\ \citenamefont {Pan}}]{PhaseprogGBS_2021}%
  \BibitemOpen
  \bibfield  {author} {\bibinfo {author} {\bibfnamefont {H.-S.}\ \bibnamefont
  {Zhong}}, \bibinfo {author} {\bibfnamefont {Y.-H.}\ \bibnamefont {Deng}},
  \bibinfo {author} {\bibfnamefont {J.}~\bibnamefont {Qin}}, \bibinfo {author}
  {\bibfnamefont {H.}~\bibnamefont {Wang}}, \bibinfo {author} {\bibfnamefont
  {M.-C.}\ \bibnamefont {Chen}}, \bibinfo {author} {\bibfnamefont {L.-C.}\
  \bibnamefont {Peng}}, \bibinfo {author} {\bibfnamefont {Y.-H.}\ \bibnamefont
  {Luo}}, \bibinfo {author} {\bibfnamefont {D.}~\bibnamefont {Wu}}, \bibinfo
  {author} {\bibfnamefont {S.-Q.}\ \bibnamefont {Gong}}, \bibinfo {author}
  {\bibfnamefont {H.}~\bibnamefont {Su}}, \bibinfo {author} {\bibfnamefont
  {Y.}~\bibnamefont {Hu}}, \bibinfo {author} {\bibfnamefont {P.}~\bibnamefont
  {Hu}}, \bibinfo {author} {\bibfnamefont {X.-Y.}\ \bibnamefont {Yang}},
  \bibinfo {author} {\bibfnamefont {W.-J.}\ \bibnamefont {Zhang}}, \bibinfo
  {author} {\bibfnamefont {H.}~\bibnamefont {Li}}, \bibinfo {author}
  {\bibfnamefont {Y.}~\bibnamefont {Li}}, \bibinfo {author} {\bibfnamefont
  {X.}~\bibnamefont {Jiang}}, \bibinfo {author} {\bibfnamefont
  {L.}~\bibnamefont {Gan}}, \bibinfo {author} {\bibfnamefont {G.}~\bibnamefont
  {Yang}}, \bibinfo {author} {\bibfnamefont {L.}~\bibnamefont {You}}, \bibinfo
  {author} {\bibfnamefont {Z.}~\bibnamefont {Wang}}, \bibinfo {author}
  {\bibfnamefont {L.}~\bibnamefont {Li}}, \bibinfo {author} {\bibfnamefont
  {N.-L.}\ \bibnamefont {Liu}}, \bibinfo {author} {\bibfnamefont {J.~J.}\
  \bibnamefont {Renema}}, \bibinfo {author} {\bibfnamefont {C.-Y.}\
  \bibnamefont {Lu}}, \ and\ \bibinfo {author} {\bibfnamefont {J.-W.}\
  \bibnamefont {Pan}},\ }\href {\doibase 10.1103/PhysRevLett.127.180502}
  {\bibfield  {journal} {\bibinfo  {journal} {Phys. Rev. Lett.}\ }\textbf
  {\bibinfo {volume} {127}},\ \bibinfo {pages} {180502} (\bibinfo {year}
  {2021})}\BibitemShut {NoStop}%
\bibitem [{\citenamefont {Madsen}\ \emph {et~al.}(2022)\citenamefont {Madsen},
  \citenamefont {Laudenbach}, \citenamefont {Askarani}, \citenamefont
  {Rortais}, \citenamefont {Vincent}, \citenamefont {Bulmer}, \citenamefont
  {Miatto}, \citenamefont {Neuhaus}, \citenamefont {Helt}, \citenamefont
  {Collins} \emph {et~al.}}]{madsen2022GBS}%
  \BibitemOpen
  \bibfield  {author} {\bibinfo {author} {\bibfnamefont {L.~S.}\ \bibnamefont
  {Madsen}}, \bibinfo {author} {\bibfnamefont {F.}~\bibnamefont {Laudenbach}},
  \bibinfo {author} {\bibfnamefont {M.~F.}\ \bibnamefont {Askarani}}, \bibinfo
  {author} {\bibfnamefont {F.}~\bibnamefont {Rortais}}, \bibinfo {author}
  {\bibfnamefont {T.}~\bibnamefont {Vincent}}, \bibinfo {author} {\bibfnamefont
  {J.~F.}\ \bibnamefont {Bulmer}}, \bibinfo {author} {\bibfnamefont {F.~M.}\
  \bibnamefont {Miatto}}, \bibinfo {author} {\bibfnamefont {L.}~\bibnamefont
  {Neuhaus}}, \bibinfo {author} {\bibfnamefont {L.~G.}\ \bibnamefont {Helt}},
  \bibinfo {author} {\bibfnamefont {M.~J.}\ \bibnamefont {Collins}},  \emph
  {et~al.},\ }\href@noop {} {\bibfield  {journal} {\bibinfo  {journal}
  {Nature}\ }\textbf {\bibinfo {volume} {606}},\ \bibinfo {pages} {75}
  (\bibinfo {year} {2022})}\BibitemShut {NoStop}%
\bibitem [{\citenamefont {Chakhmakhchyan}\ and\ \citenamefont
  {Cerf}(2017)}]{levonBS_Gaussian_meas}%
  \BibitemOpen
  \bibfield  {author} {\bibinfo {author} {\bibfnamefont {L.}~\bibnamefont
  {Chakhmakhchyan}}\ and\ \bibinfo {author} {\bibfnamefont {N.~J.}\
  \bibnamefont {Cerf}},\ }\href@noop {} {\bibfield  {journal} {\bibinfo
  {journal} {Physical Review A}\ }\textbf {\bibinfo {volume} {96}},\ \bibinfo
  {pages} {032326} (\bibinfo {year} {2017})}\BibitemShut {NoStop}%
\bibitem [{\citenamefont {Chabaud}\ \emph {et~al.}(2017)\citenamefont
  {Chabaud}, \citenamefont {Douce}, \citenamefont {Markham}, \citenamefont {van
  Loock}, \citenamefont {Kashefi},\ and\ \citenamefont
  {Ferrini}}]{chabaud2017photonadded}%
  \BibitemOpen
  \bibfield  {author} {\bibinfo {author} {\bibfnamefont {U.}~\bibnamefont
  {Chabaud}}, \bibinfo {author} {\bibfnamefont {T.}~\bibnamefont {Douce}},
  \bibinfo {author} {\bibfnamefont {D.}~\bibnamefont {Markham}}, \bibinfo
  {author} {\bibfnamefont {P.}~\bibnamefont {van Loock}}, \bibinfo {author}
  {\bibfnamefont {E.}~\bibnamefont {Kashefi}}, \ and\ \bibinfo {author}
  {\bibfnamefont {G.}~\bibnamefont {Ferrini}},\ }\href@noop {} {\bibfield
  {journal} {\bibinfo  {journal} {Physical Review A}\ }\textbf {\bibinfo
  {volume} {96}},\ \bibinfo {pages} {062307} (\bibinfo {year}
  {2017})}\BibitemShut {NoStop}%
\bibitem [{\citenamefont {Bremner}\ \emph {et~al.}(2016)\citenamefont
  {Bremner}, \citenamefont {Montanaro},\ and\ \citenamefont
  {Shepherd}}]{bremner2016IQP}%
  \BibitemOpen
  \bibfield  {author} {\bibinfo {author} {\bibfnamefont {M.~J.}\ \bibnamefont
  {Bremner}}, \bibinfo {author} {\bibfnamefont {A.}~\bibnamefont {Montanaro}},
  \ and\ \bibinfo {author} {\bibfnamefont {D.~J.}\ \bibnamefont {Shepherd}},\
  }\href@noop {} {\bibfield  {journal} {\bibinfo  {journal} {Physical review
  letters}\ }\textbf {\bibinfo {volume} {117}},\ \bibinfo {pages} {080501}
  (\bibinfo {year} {2016})}\BibitemShut {NoStop}%
\bibitem [{\citenamefont {Douce}\ \emph {et~al.}(2017)\citenamefont {Douce},
  \citenamefont {Markham}, \citenamefont {Kashefi}, \citenamefont {Diamanti},
  \citenamefont {Coudreau}, \citenamefont {Milman}, \citenamefont {van Loock},\
  and\ \citenamefont {Ferrini}}]{douce2017continuous}%
  \BibitemOpen
  \bibfield  {author} {\bibinfo {author} {\bibfnamefont {T.}~\bibnamefont
  {Douce}}, \bibinfo {author} {\bibfnamefont {D.}~\bibnamefont {Markham}},
  \bibinfo {author} {\bibfnamefont {E.}~\bibnamefont {Kashefi}}, \bibinfo
  {author} {\bibfnamefont {E.}~\bibnamefont {Diamanti}}, \bibinfo {author}
  {\bibfnamefont {T.}~\bibnamefont {Coudreau}}, \bibinfo {author}
  {\bibfnamefont {P.}~\bibnamefont {Milman}}, \bibinfo {author} {\bibfnamefont
  {P.}~\bibnamefont {van Loock}}, \ and\ \bibinfo {author} {\bibfnamefont
  {G.}~\bibnamefont {Ferrini}},\ }\href@noop {} {\bibfield  {journal} {\bibinfo
   {journal} {Physical review letters}\ }\textbf {\bibinfo {volume} {118}},\
  \bibinfo {pages} {070503} (\bibinfo {year} {2017})}\BibitemShut {NoStop}%
\bibitem [{\citenamefont {Douce}\ \emph {et~al.}(2019)\citenamefont {Douce},
  \citenamefont {Markham}, \citenamefont {Kashefi}, \citenamefont {van Loock},\
  and\ \citenamefont {Ferrini}}]{douce2019probabilistic}%
  \BibitemOpen
  \bibfield  {author} {\bibinfo {author} {\bibfnamefont {T.}~\bibnamefont
  {Douce}}, \bibinfo {author} {\bibfnamefont {D.}~\bibnamefont {Markham}},
  \bibinfo {author} {\bibfnamefont {E.}~\bibnamefont {Kashefi}}, \bibinfo
  {author} {\bibfnamefont {P.}~\bibnamefont {van Loock}}, \ and\ \bibinfo
  {author} {\bibfnamefont {G.}~\bibnamefont {Ferrini}},\ }\href@noop {}
  {\bibfield  {journal} {\bibinfo  {journal} {Physical Review A}\ }\textbf
  {\bibinfo {volume} {99}},\ \bibinfo {pages} {012344} (\bibinfo {year}
  {2019})}\BibitemShut {NoStop}%
\bibitem [{\citenamefont {Spagnolo}\ \emph {et~al.}(2023)\citenamefont
  {Spagnolo}, \citenamefont {Brod}, \citenamefont {Galv{\~a}o},\ and\
  \citenamefont {Sciarrino}}]{nonlinear_BS}%
  \BibitemOpen
  \bibfield  {author} {\bibinfo {author} {\bibfnamefont {N.}~\bibnamefont
  {Spagnolo}}, \bibinfo {author} {\bibfnamefont {D.~J.}\ \bibnamefont {Brod}},
  \bibinfo {author} {\bibfnamefont {E.~F.}\ \bibnamefont {Galv{\~a}o}}, \ and\
  \bibinfo {author} {\bibfnamefont {F.}~\bibnamefont {Sciarrino}},\ }\href@noop
  {} {\bibfield  {journal} {\bibinfo  {journal} {npj Quantum Information}\
  }\textbf {\bibinfo {volume} {9}},\ \bibinfo {pages} {3} (\bibinfo {year}
  {2023})}\BibitemShut {NoStop}%
\bibitem [{\citenamefont {Kirchmair}\ \emph {et~al.}(2013)\citenamefont
  {Kirchmair}, \citenamefont {Vlastakis}, \citenamefont {Leghtas},
  \citenamefont {Nigg}, \citenamefont {Paik}, \citenamefont {Ginossar},
  \citenamefont {Mirrahimi},\ and\ \citenamefont
  {Frunzio}}]{exp_kirchmair2013observation}%
  \BibitemOpen
  \bibfield  {author} {\bibinfo {author} {\bibfnamefont {G.}~\bibnamefont
  {Kirchmair}}, \bibinfo {author} {\bibfnamefont {B.}~\bibnamefont
  {Vlastakis}}, \bibinfo {author} {\bibfnamefont {Z.}~\bibnamefont {Leghtas}},
  \bibinfo {author} {\bibfnamefont {S.~E.}\ \bibnamefont {Nigg}}, \bibinfo
  {author} {\bibfnamefont {H.}~\bibnamefont {Paik}}, \bibinfo {author}
  {\bibfnamefont {E.}~\bibnamefont {Ginossar}}, \bibinfo {author}
  {\bibfnamefont {M.}~\bibnamefont {Mirrahimi}}, \ and\ \bibinfo {author}
  {\bibfnamefont {L.}~\bibnamefont {Frunzio}},\ }\href@noop {} {\bibfield
  {journal} {\bibinfo  {journal} {Nature}\ }\textbf {\bibinfo {volume} {495}},\
  \bibinfo {pages} {205} (\bibinfo {year} {2013})}\BibitemShut {NoStop}%
\bibitem [{\citenamefont {Chang}\ \emph {et~al.}(2014)\citenamefont {Chang},
  \citenamefont {Vuleti{\'c}},\ and\ \citenamefont
  {Lukin}}]{exp_chang2014quantum}%
  \BibitemOpen
  \bibfield  {author} {\bibinfo {author} {\bibfnamefont {D.~E.}\ \bibnamefont
  {Chang}}, \bibinfo {author} {\bibfnamefont {V.}~\bibnamefont {Vuleti{\'c}}},
  \ and\ \bibinfo {author} {\bibfnamefont {M.~D.}\ \bibnamefont {Lukin}},\
  }\href@noop {} {\bibfield  {journal} {\bibinfo  {journal} {Nature Photonics}\
  }\textbf {\bibinfo {volume} {8}},\ \bibinfo {pages} {685} (\bibinfo {year}
  {2014})}\BibitemShut {NoStop}%
\bibitem [{\citenamefont {Hacker}\ \emph {et~al.}(2016)\citenamefont {Hacker},
  \citenamefont {Welte}, \citenamefont {Rempe},\ and\ \citenamefont
  {Ritter}}]{exp_hacker2016photon}%
  \BibitemOpen
  \bibfield  {author} {\bibinfo {author} {\bibfnamefont {B.}~\bibnamefont
  {Hacker}}, \bibinfo {author} {\bibfnamefont {S.}~\bibnamefont {Welte}},
  \bibinfo {author} {\bibfnamefont {G.}~\bibnamefont {Rempe}}, \ and\ \bibinfo
  {author} {\bibfnamefont {S.}~\bibnamefont {Ritter}},\ }\href@noop {}
  {\bibfield  {journal} {\bibinfo  {journal} {Nature}\ }\textbf {\bibinfo
  {volume} {536}},\ \bibinfo {pages} {193} (\bibinfo {year}
  {2016})}\BibitemShut {NoStop}%
\bibitem [{\citenamefont {Tiarks}\ \emph {et~al.}(2019)\citenamefont {Tiarks},
  \citenamefont {Schmidt-Eberle}, \citenamefont {Stolz}, \citenamefont
  {Rempe},\ and\ \citenamefont {D{\"u}rr}}]{exp_tiarks2019photon}%
  \BibitemOpen
  \bibfield  {author} {\bibinfo {author} {\bibfnamefont {D.}~\bibnamefont
  {Tiarks}}, \bibinfo {author} {\bibfnamefont {S.}~\bibnamefont
  {Schmidt-Eberle}}, \bibinfo {author} {\bibfnamefont {T.}~\bibnamefont
  {Stolz}}, \bibinfo {author} {\bibfnamefont {G.}~\bibnamefont {Rempe}}, \ and\
  \bibinfo {author} {\bibfnamefont {S.}~\bibnamefont {D{\"u}rr}},\ }\href@noop
  {} {\bibfield  {journal} {\bibinfo  {journal} {Nature Physics}\ }\textbf
  {\bibinfo {volume} {15}},\ \bibinfo {pages} {124} (\bibinfo {year}
  {2019})}\BibitemShut {NoStop}%
\bibitem [{\citenamefont {Moreno-Cardoner}\ \emph {et~al.}(2021)\citenamefont
  {Moreno-Cardoner}, \citenamefont {Goncalves},\ and\ \citenamefont
  {Chang}}]{exp_moreno2021quantum}%
  \BibitemOpen
  \bibfield  {author} {\bibinfo {author} {\bibfnamefont {M.}~\bibnamefont
  {Moreno-Cardoner}}, \bibinfo {author} {\bibfnamefont {D.}~\bibnamefont
  {Goncalves}}, \ and\ \bibinfo {author} {\bibfnamefont {D.~E.}\ \bibnamefont
  {Chang}},\ }\href@noop {} {\bibfield  {journal} {\bibinfo  {journal}
  {Physical Review Letters}\ }\textbf {\bibinfo {volume} {127}},\ \bibinfo
  {pages} {263602} (\bibinfo {year} {2021})}\BibitemShut {NoStop}%
\bibitem [{\citenamefont {Reuer}\ \emph {et~al.}(2022)\citenamefont {Reuer},
  \citenamefont {Besse}, \citenamefont {Wernli}, \citenamefont {Magnard},
  \citenamefont {Kurpiers}, \citenamefont {Norris}, \citenamefont {Wallraff},\
  and\ \citenamefont {Eichler}}]{exp_reuer2022realization}%
  \BibitemOpen
  \bibfield  {author} {\bibinfo {author} {\bibfnamefont {K.}~\bibnamefont
  {Reuer}}, \bibinfo {author} {\bibfnamefont {J.-C.}\ \bibnamefont {Besse}},
  \bibinfo {author} {\bibfnamefont {L.}~\bibnamefont {Wernli}}, \bibinfo
  {author} {\bibfnamefont {P.}~\bibnamefont {Magnard}}, \bibinfo {author}
  {\bibfnamefont {P.}~\bibnamefont {Kurpiers}}, \bibinfo {author}
  {\bibfnamefont {G.~J.}\ \bibnamefont {Norris}}, \bibinfo {author}
  {\bibfnamefont {A.}~\bibnamefont {Wallraff}}, \ and\ \bibinfo {author}
  {\bibfnamefont {C.}~\bibnamefont {Eichler}},\ }\href@noop {} {\bibfield
  {journal} {\bibinfo  {journal} {Physical Review X}\ }\textbf {\bibinfo
  {volume} {12}},\ \bibinfo {pages} {011008} (\bibinfo {year}
  {2022})}\BibitemShut {NoStop}%
\bibitem [{\citenamefont {Peropadre}\ \emph {et~al.}(2016)\citenamefont
  {Peropadre}, \citenamefont {Guerreschi}, \citenamefont {Huh},\ and\
  \citenamefont {Aspuru-Guzik}}]{peropadre2016proposal}%
  \BibitemOpen
  \bibfield  {author} {\bibinfo {author} {\bibfnamefont {B.}~\bibnamefont
  {Peropadre}}, \bibinfo {author} {\bibfnamefont {G.~G.}\ \bibnamefont
  {Guerreschi}}, \bibinfo {author} {\bibfnamefont {J.}~\bibnamefont {Huh}}, \
  and\ \bibinfo {author} {\bibfnamefont {A.}~\bibnamefont {Aspuru-Guzik}},\
  }\href@noop {} {\bibfield  {journal} {\bibinfo  {journal} {Physical review
  letters}\ }\textbf {\bibinfo {volume} {117}},\ \bibinfo {pages} {140505}
  (\bibinfo {year} {2016})}\BibitemShut {NoStop}%
\bibitem [{\citenamefont {Young}\ \emph {et~al.}(2024)\citenamefont {Young},
  \citenamefont {Geller}, \citenamefont {Eckner}, \citenamefont {Schine},
  \citenamefont {Glancy}, \citenamefont {Knill},\ and\ \citenamefont
  {Kaufman}}]{young2024atomic}%
  \BibitemOpen
  \bibfield  {author} {\bibinfo {author} {\bibfnamefont {A.~W.}\ \bibnamefont
  {Young}}, \bibinfo {author} {\bibfnamefont {S.}~\bibnamefont {Geller}},
  \bibinfo {author} {\bibfnamefont {W.~J.}\ \bibnamefont {Eckner}}, \bibinfo
  {author} {\bibfnamefont {N.}~\bibnamefont {Schine}}, \bibinfo {author}
  {\bibfnamefont {S.}~\bibnamefont {Glancy}}, \bibinfo {author} {\bibfnamefont
  {E.}~\bibnamefont {Knill}}, \ and\ \bibinfo {author} {\bibfnamefont {A.~M.}\
  \bibnamefont {Kaufman}},\ }\href@noop {} {\bibfield  {journal} {\bibinfo
  {journal} {Nature}\ }\textbf {\bibinfo {volume} {629}},\ \bibinfo {pages}
  {311} (\bibinfo {year} {2024})}\BibitemShut {NoStop}%
\bibitem [{\citenamefont {Maskara}\ \emph {et~al.}(2022)\citenamefont
  {Maskara}, \citenamefont {Deshpande}, \citenamefont {Ehrenberg},
  \citenamefont {Tran}, \citenamefont {Fefferman},\ and\ \citenamefont
  {Gorshkov}}]{maskara2022complexity}%
  \BibitemOpen
  \bibfield  {author} {\bibinfo {author} {\bibfnamefont {N.}~\bibnamefont
  {Maskara}}, \bibinfo {author} {\bibfnamefont {A.}~\bibnamefont {Deshpande}},
  \bibinfo {author} {\bibfnamefont {A.}~\bibnamefont {Ehrenberg}}, \bibinfo
  {author} {\bibfnamefont {M.~C.}\ \bibnamefont {Tran}}, \bibinfo {author}
  {\bibfnamefont {B.}~\bibnamefont {Fefferman}}, \ and\ \bibinfo {author}
  {\bibfnamefont {A.~V.}\ \bibnamefont {Gorshkov}},\ }\href@noop {} {\bibfield
  {journal} {\bibinfo  {journal} {Physical Review Letters}\ }\textbf {\bibinfo
  {volume} {129}},\ \bibinfo {pages} {150604} (\bibinfo {year}
  {2022})}\BibitemShut {NoStop}%
\bibitem [{\citenamefont {Somma}\ \emph {et~al.}(2002)\citenamefont {Somma},
  \citenamefont {Ortiz}, \citenamefont {Gubernatis}, \citenamefont {Knill},\
  and\ \citenamefont {Laflamme}}]{somma2002}%
  \BibitemOpen
  \bibfield  {author} {\bibinfo {author} {\bibfnamefont {R.}~\bibnamefont
  {Somma}}, \bibinfo {author} {\bibfnamefont {G.}~\bibnamefont {Ortiz}},
  \bibinfo {author} {\bibfnamefont {J.~E.}\ \bibnamefont {Gubernatis}},
  \bibinfo {author} {\bibfnamefont {E.}~\bibnamefont {Knill}}, \ and\ \bibinfo
  {author} {\bibfnamefont {R.}~\bibnamefont {Laflamme}},\ }\href {\doibase
  10.1103/PhysRevA.65.042323} {\bibfield  {journal} {\bibinfo  {journal} {Phys.
  Rev. A}\ }\textbf {\bibinfo {volume} {65}},\ \bibinfo {pages} {042323}
  (\bibinfo {year} {2002})}\BibitemShut {NoStop}%
\bibitem [{\citenamefont {Somma}(2019)}]{somma2019}%
  \BibitemOpen
  \bibfield  {author} {\bibinfo {author} {\bibfnamefont {R.~D.}\ \bibnamefont
  {Somma}},\ }\href@noop {} {\bibfield  {journal} {\bibinfo  {journal} {New
  Journal of Physics}\ }\textbf {\bibinfo {volume} {21}},\ \bibinfo {pages}
  {123025} (\bibinfo {year} {2019})}\BibitemShut {NoStop}%
\bibitem [{\citenamefont {Cleve}\ \emph {et~al.}(1998)\citenamefont {Cleve},
  \citenamefont {Ekert}, \citenamefont {Macchiavello},\ and\ \citenamefont
  {Mosca}}]{Cleve1998}%
  \BibitemOpen
  \bibfield  {author} {\bibinfo {author} {\bibfnamefont {R.}~\bibnamefont
  {Cleve}}, \bibinfo {author} {\bibfnamefont {A.}~\bibnamefont {Ekert}},
  \bibinfo {author} {\bibfnamefont {C.}~\bibnamefont {Macchiavello}}, \ and\
  \bibinfo {author} {\bibfnamefont {M.}~\bibnamefont {Mosca}},\ }\href
  {\doibase 10.1098/rspa.1998.0164} {\bibfield  {journal} {\bibinfo  {journal}
  {Proceedings of the Royal Society of London. Series A: Mathematical, Physical
  and Engineering Sciences}\ }\textbf {\bibinfo {volume} {454}},\ \bibinfo
  {pages} {339} (\bibinfo {year} {1998})}\BibitemShut {NoStop}%
\bibitem [{\citenamefont {Kruse}\ \emph {et~al.}(2019)\citenamefont {Kruse},
  \citenamefont {Hamilton}, \citenamefont {Sansoni}, \citenamefont {Barkhofen},
  \citenamefont {Silberhorn},\ and\ \citenamefont {Jex}}]{kruse2019detailed}%
  \BibitemOpen
  \bibfield  {author} {\bibinfo {author} {\bibfnamefont {R.}~\bibnamefont
  {Kruse}}, \bibinfo {author} {\bibfnamefont {C.~S.}\ \bibnamefont {Hamilton}},
  \bibinfo {author} {\bibfnamefont {L.}~\bibnamefont {Sansoni}}, \bibinfo
  {author} {\bibfnamefont {S.}~\bibnamefont {Barkhofen}}, \bibinfo {author}
  {\bibfnamefont {C.}~\bibnamefont {Silberhorn}}, \ and\ \bibinfo {author}
  {\bibfnamefont {I.}~\bibnamefont {Jex}},\ }\href@noop {} {\bibfield
  {journal} {\bibinfo  {journal} {Physical Review A}\ }\textbf {\bibinfo
  {volume} {100}},\ \bibinfo {pages} {032326} (\bibinfo {year}
  {2019})}\BibitemShut {NoStop}%
\bibitem [{\citenamefont {Dodonov}\ \emph {et~al.}(1974)\citenamefont
  {Dodonov}, \citenamefont {Malkin},\ and\ \citenamefont
  {Man'Ko}}]{dodonov1974even}%
  \BibitemOpen
  \bibfield  {author} {\bibinfo {author} {\bibfnamefont {V.}~\bibnamefont
  {Dodonov}}, \bibinfo {author} {\bibfnamefont {I.}~\bibnamefont {Malkin}}, \
  and\ \bibinfo {author} {\bibfnamefont {V.}~\bibnamefont {Man'Ko}},\
  }\href@noop {} {\bibfield  {journal} {\bibinfo  {journal} {Physica}\ }\textbf
  {\bibinfo {volume} {72}},\ \bibinfo {pages} {597} (\bibinfo {year}
  {1974})}\BibitemShut {NoStop}%
\bibitem [{\citenamefont {Daley}\ \emph {et~al.}(2012)\citenamefont {Daley},
  \citenamefont {Pichler}, \citenamefont {Schachenmayer},\ and\ \citenamefont
  {Zoller}}]{daley2012measuring}%
  \BibitemOpen
  \bibfield  {author} {\bibinfo {author} {\bibfnamefont {A.~J.}\ \bibnamefont
  {Daley}}, \bibinfo {author} {\bibfnamefont {H.}~\bibnamefont {Pichler}},
  \bibinfo {author} {\bibfnamefont {J.}~\bibnamefont {Schachenmayer}}, \ and\
  \bibinfo {author} {\bibfnamefont {P.}~\bibnamefont {Zoller}},\ }\href@noop {}
  {\bibfield  {journal} {\bibinfo  {journal} {Physical review letters}\
  }\textbf {\bibinfo {volume} {109}},\ \bibinfo {pages} {020505} (\bibinfo
  {year} {2012})}\BibitemShut {NoStop}%
\bibitem [{\citenamefont {Filip}(2002)}]{Filip2002}%
  \BibitemOpen
  \bibfield  {author} {\bibinfo {author} {\bibfnamefont {R.}~\bibnamefont
  {Filip}},\ }\href {\doibase 10.1103/PhysRevA.65.062320} {\bibfield  {journal}
  {\bibinfo  {journal} {Phys. Rev. A}\ }\textbf {\bibinfo {volume} {65}},\
  \bibinfo {pages} {062320} (\bibinfo {year} {2002})}\BibitemShut {NoStop}%
\bibitem [{\citenamefont {Kounalakis}\ \emph {et~al.}(2018)\citenamefont
  {Kounalakis}, \citenamefont {Dickel}, \citenamefont {Bruno}, \citenamefont
  {Langford},\ and\ \citenamefont {Steele}}]{crossKerrcQED2018}%
  \BibitemOpen
  \bibfield  {author} {\bibinfo {author} {\bibfnamefont {M.}~\bibnamefont
  {Kounalakis}}, \bibinfo {author} {\bibfnamefont {C.}~\bibnamefont {Dickel}},
  \bibinfo {author} {\bibfnamefont {A.}~\bibnamefont {Bruno}}, \bibinfo
  {author} {\bibfnamefont {N.}~\bibnamefont {Langford}}, \ and\ \bibinfo
  {author} {\bibfnamefont {G.}~\bibnamefont {Steele}},\ }\href@noop {}
  {\bibfield  {journal} {\bibinfo  {journal} {npj Quantum Information}\
  }\textbf {\bibinfo {volume} {4}},\ \bibinfo {pages} {38} (\bibinfo {year}
  {2018})}\BibitemShut {NoStop}%
\bibitem [{\citenamefont {Huh}\ \emph {et~al.}(2015)\citenamefont {Huh},
  \citenamefont {Guerreschi}, \citenamefont {Peropadre}, \citenamefont
  {McClean},\ and\ \citenamefont {Aspuru-Guzik}}]{huh2015boson}%
  \BibitemOpen
  \bibfield  {author} {\bibinfo {author} {\bibfnamefont {J.}~\bibnamefont
  {Huh}}, \bibinfo {author} {\bibfnamefont {G.~G.}\ \bibnamefont {Guerreschi}},
  \bibinfo {author} {\bibfnamefont {B.}~\bibnamefont {Peropadre}}, \bibinfo
  {author} {\bibfnamefont {J.~R.}\ \bibnamefont {McClean}}, \ and\ \bibinfo
  {author} {\bibfnamefont {A.}~\bibnamefont {Aspuru-Guzik}},\ }\href@noop {}
  {\bibfield  {journal} {\bibinfo  {journal} {Nature Photonics}\ }\textbf
  {\bibinfo {volume} {9}},\ \bibinfo {pages} {615} (\bibinfo {year}
  {2015})}\BibitemShut {NoStop}%
\bibitem [{\citenamefont {Oh}\ \emph {et~al.}(2024)\citenamefont {Oh},
  \citenamefont {Lim}, \citenamefont {Wong}, \citenamefont {Fefferman},\ and\
  \citenamefont {Jiang}}]{oh2024quantum}%
  \BibitemOpen
  \bibfield  {author} {\bibinfo {author} {\bibfnamefont {C.}~\bibnamefont
  {Oh}}, \bibinfo {author} {\bibfnamefont {Y.}~\bibnamefont {Lim}}, \bibinfo
  {author} {\bibfnamefont {Y.}~\bibnamefont {Wong}}, \bibinfo {author}
  {\bibfnamefont {B.}~\bibnamefont {Fefferman}}, \ and\ \bibinfo {author}
  {\bibfnamefont {L.}~\bibnamefont {Jiang}},\ }\href@noop {} {\bibfield
  {journal} {\bibinfo  {journal} {Nature Physics}\ }\textbf {\bibinfo {volume}
  {20}},\ \bibinfo {pages} {225} (\bibinfo {year} {2024})}\BibitemShut
  {NoStop}%
\bibitem [{\citenamefont {Chabaud}\ \emph {et~al.}(2021)\citenamefont
  {Chabaud}, \citenamefont {Markham},\ and\ \citenamefont
  {Sohbi}}]{chabaud2021quantum}%
  \BibitemOpen
  \bibfield  {author} {\bibinfo {author} {\bibfnamefont {U.}~\bibnamefont
  {Chabaud}}, \bibinfo {author} {\bibfnamefont {D.}~\bibnamefont {Markham}}, \
  and\ \bibinfo {author} {\bibfnamefont {A.}~\bibnamefont {Sohbi}},\
  }\href@noop {} {\bibfield  {journal} {\bibinfo  {journal} {Quantum}\ }\textbf
  {\bibinfo {volume} {5}},\ \bibinfo {pages} {496} (\bibinfo {year}
  {2021})}\BibitemShut {NoStop}%
\bibitem [{\citenamefont {Marshall}\ \emph {et~al.}(2011)\citenamefont
  {Marshall}, \citenamefont {Olkin},\ and\ \citenamefont
  {Arnold}}]{Majorization}%
  \BibitemOpen
  \bibfield  {author} {\bibinfo {author} {\bibfnamefont {A.~W.}\ \bibnamefont
  {Marshall}}, \bibinfo {author} {\bibfnamefont {I.}~\bibnamefont {Olkin}}, \
  and\ \bibinfo {author} {\bibfnamefont {B.~C.}\ \bibnamefont {Arnold}},\
  }\href {\doibase 10.1007/978-0-387-68276-1} {\emph {\bibinfo {title}
  {Inequalities: Theory of Majorization and its Applications}}},\ \bibinfo
  {edition} {2nd}\ ed.,\ Vol.\ \bibinfo {volume} {143}\ (\bibinfo  {publisher}
  {Springer},\ \bibinfo {year} {2011})\BibitemShut {NoStop}%
\end{thebibliography}%
\onecolumngrid

\appendix

\section{Construction of the non-linear Hamiltonian\label{app:lem1}}

\begin{proof}[Proof of Lemma~\ref{lemma:non_degenerate}]
Let us consider a \emph{collision free} outcome $\ket{S}$, containing $N$ photons, where each mode contains at most one photon. Without loss of generality, we assume that the photons are contained in the first $N$ modes, i.e., 
\begin{equation}
    \ket{S}= \ket{1}^{\otimes N}\ket{0}^{\otimes (M-N)}. 
\end{equation}
We can construct a Hamiltonian for which this state is an eigenstate with a non-degenerate eigenvalue. To do this, define the Hamiltonian
\begin{equation}
    \hat{H}^{(1)}_{NL} = \sum_{j=1}^N \hat{n}_j + \gamma \sum_{j=1}^N \hat{n}_j^2 + \Gamma \sum_{j=N+1}^M  \hat{n}_j.
\end{equation}
We have that
\begin{equation}
    \hat{H}^{(1)}_{NL} \ket{S} = (1+\gamma) N \ket{S}.
\end{equation}
We are now going to choose the parameters $\gamma$ and $\Gamma$ so that any eigenstate of $\hat{H}^{(1)}_{NL}$ other than $\ket{S}$ is associated with an eigenvalue that is stricly different from $(1+\gamma)N$. 

Since  $\hat{H}^{(1)}_{NL}$ is a linear combination of terms $\hat{n}_j$ and $\hat{n}_j^2$, any product of Fock states is an eigenstate of this Hamiltonian. We first focus on eigenstates with zero photons in the last $M-N$ modes and whose total number of photons is $k$. Denote any such state by
\begin{equation} \label{eq:state2}
    \ket{n_1, n_2, \cdots, n_n} \ket{0}^{\otimes (M-N)}, \quad \mathrm{with}~~ \, \sum_{j=1}^N n_j = k,
\end{equation}
and define the vector $\boldsymbol{n}^{(k)} = (n_1, n_2, \cdots, n_n)$.  Via an appropriate choice of $\gamma$, we can ensure that $\ket{S}$ is the only eigenstate with eigenvalue $(1+\gamma) N$ among this set of states. To see this, note that the eigenvalue associated with such $k-$photon state is given by
\begin{equation}
    \begin{aligned}
        \hat{H}^{(1)}_{NL}\ket{\boldsymbol{n}^{(k)}} \ket{0}^{\otimes (M-N)} & = E_{\boldsymbol{n}^{(k)}} \ket{\boldsymbol{n}^{(k)}} \ket{0}^{\otimes (M-N)} \\
        & = [k+\gamma F(\boldsymbol{n}^{(k)})] \ket{\boldsymbol{n}^{(k)}} \ket{0}^{\otimes (M-N)},
    \end{aligned}
\end{equation}
where we defined the function
\begin{equation}
    F(\boldsymbol{n}^{(k)}) = \sum_{j=1}^N n_j^2.
\end{equation}
In order to bound the range of values that the function $F$ can take, it is useful to note that this function is strictly Schur-convex~\cite{Majorization}. By definition, a strictly Schur-convex function $F$ is such that, if a vector $\boldsymbol{n}$ is majorized by a vector $\boldsymbol{m}$ (written $\boldsymbol{n} \prec \boldsymbol{m}$)~\cite{Majorization}, but $\boldsymbol{n}$ is not a permutation of $\boldsymbol{m}$, then $F(\boldsymbol{n}) < F(\boldsymbol{m})$. Hence, for a given $k$, the minimum of $F$ is obtained for the $N$-dimensional vector $ \boldsymbol{1}^{(k)} = (1,1, \dots,1,0, \dots 0)$, containing $k$ $1$'s and $(N-k)$ $0$'s, since this vector is majorized by any vector $\boldsymbol{n}^{(k)}$ which is not a permutation of $ \boldsymbol{1}^{(k)}$. In turn, any vector $\boldsymbol{n}^{(k)}$ is majorized by the $N$-dimensional vector $\boldsymbol{\delta}^{(k)} = (k,0,0,\cdots,0)$, unless it is obtained from a permutation of the entries of this vector. In summary, this implies that for any vector $\boldsymbol{n}^{(k)}$
\begin{equation}
    k= F(\boldsymbol{1}^{(k)}) \leq F(\boldsymbol{n}^{(k)}) \leq F(\boldsymbol{\delta}^{(k)})= k^2,  
\end{equation}
where the equality on the left (right) is achieved if and only if $\boldsymbol{n}^{(k)}$ is obtained by a permutation of the entries of $\boldsymbol{1}^{(k)}$ ($\boldsymbol{\delta}^{(k)}$). In particular, if $k=N$, the vector $\boldsymbol{1}^{(N)}$ is the only vector that minimizes $F$ and, for any other vector $\boldsymbol{n}^{(N)}$, $F(\boldsymbol{n}^{(N)})\geq  (N+3)$. 

To guarantee that the vector $\boldsymbol{1}^{(N)}$ is, in fact, the only vector for which $E_{\boldsymbol{1}^{(N)}}=(1+\gamma)N$, we impose the following conditions
\begin{align}
   (i)~ E_{\boldsymbol{n}^{(k)}}< (1+\gamma)N, ~~~~\mathrm{for~~}k\leq N-1;\\
    (ii)~E_{\boldsymbol{n}^{(k)}}> (1+\gamma)N,~~~~ \mathrm{for~~} k\geq N+1.
\end{align}
Note that, for $k\leq N-1$, we have  
\begin{equation}
    E_{\boldsymbol{n}^{(k)}}= k+\gamma F(\boldsymbol{n}^{(k)})\leq k+ \gamma k^2\leq (N-1)+\gamma(N-1)^2. 
\end{equation}
Hence, we can satisfy condition~$(i)$ by ensuring that, for any positive integer $N$,
\begin{equation}
(N-1) + \gamma (N-1)^2 < (1+\gamma) N,
\end{equation}
which leads to the condition
\begin{equation}
    \gamma < \frac{1}{(N-1)^2-N}.
\end{equation}
For $k\geq N+1$, we have that 
\begin{equation}
    E_{\boldsymbol{n}^{(k)}}\geq k (\gamma +1) \geq (N+1)(\gamma+1),  
\end{equation}
which means that condition~$(ii)$ is automatically ensured for any positive $\gamma$. For simplicity, we choose $\gamma= 1/N^2$.

Consider now eigenstates with at least one photon in the last $(M-N)$ modes. They will all correspond to eigenvalues strictly greater than $(1+\gamma) N$ if $\Gamma > (1+\gamma) N$, so that we can simply choose $\Gamma = N+2$. We end up with
\begin{equation}
    \hat{H}^{(1)}_{NL} = \sum_{j=1}^N \hat{n}_j + \frac{1}{N^2} \sum_{j=1}^N \hat{n}_j^2 + (N+2) \sum_{j=N+1}^M  \hat{n}_j.
\end{equation}
Finally, we rescale the full Hamiltonian $\hat{H}^{(1)}_{NL}$ so that all eigenvalues are integer-valued and define the Hamiltonian
\begin{equation}
    \hat{H}_{NL} = N^2 \sum_{j=1}^N \hat{n}_j + \sum_{j=1}^N \hat{n}_j^2 + N^2(N+2) \sum_{j=N+1}^M  \hat{n}_j.
\end{equation}
\end{proof}

\section{Bounding the tail of the \texorpdfstring{$K$}{K}-squeezers distribution.\label{app:lem2}}

\begin{proof}[Proof of Lemma~\ref{lem:truncation}]
Consider the product of $K$ single-mode squeezed states $\ket{\Psi^K(r)} = \ket{\psi_{\mathrm{SMS}}(r)}^{\otimes K}$. The probability of having a total of $n$ photon pair events is given by~\cite{Hamilton2017}
\begin{equation}
    \begin{aligned}
        p^{(K)}_n(r) & = \sum_{\substack{(n_1, n_2, \cdots, n_K) \in \mathbb{N}_0^{K} \\ n_1+n_2+\cdots+n_K=n}} \left| \bk{2 n_1, 2 n_2, \cdots, 2 n_K}{\Psi^K(r)} \right|^2 \\
        & = \binom{\frac{K}{2} + n - 1}{n} \sech^K(r) \tanh^{2n}(r).
    \end{aligned}
\end{equation}
Let $\hat{N}$ be the total photon-number operator and consider the moment generating function of the total photon-number probability distribution given by 
\begin{equation}
    \begin{aligned}
        \bra{\Psi^K(r)} \e^{\tau \hat{N}} \ket{\Psi^K(r)} & = \sum_{n=0}^{\infty} \e^{2 \tau n}p^{(K)}_n(r).
            \end{aligned}
\end{equation}
For values of $\tau$ such that $1< e^{2\tau}< 1/\tanh^2(r)$, the infinite series converges to the value
 \begin{equation}
    \begin{aligned}       
        \bra{\Psi^K(r)} \e^{\tau \hat{N}} \ket{\Psi^K(r)}
        & = \left(\frac{\sech(r)}{\sqrt{1-\e^{2\tau} \tanh^2(r)}}\right)^K.
    \end{aligned}
\end{equation}
Using the Chernoff bound, we have that the probability of having a number of photons larger than $N^*$ is bounded as 
\begin{equation}\label{eq:chernoff}
    \begin{aligned}
        \mathbb{P}(N'>N^*) & \leq \left(\frac{\sech(r)}{\sqrt{1-\e^{2\tau} \tanh^2(r)}}\right)^K e^{-\tau N^*}, \quad \forall \, \tau ~\text{s.t.}~1< e^{2\tau}< 1/\tanh^2(r).
    \end{aligned}
\end{equation}
It will be useful to reparametrize $e^{2\tau}$ in the following way
\begin{align}
    e^{2\tau}&= 1+ \beta (\tanh^{-2}(r)-1)\\&= 1+ \frac{\beta}{\sinh^2(r)},~\text{with}~0<\beta<1. 
\end{align}
This allows us to rewrite Eq.~\eqref{eq:chernoff} as 
\begin{align}\label{eq:errorbound1}
    \mathbb{P}(N'>N^*) & \leq (1-\beta)^{-K/2} \left(1+ \frac{\beta}{\sinh^2(r)}\right)^{-N^*/2},~\text{with}~0<\beta<1.
\end{align}
We would like to find a value of $N^*$ such that 
\begin{equation}   \label{eq:errorcondition}
 \mathbb{P}(N'>N^*)\leq \exp(- c N\log{N}). 
\end{equation}
To do so, let us choose $\beta=1/2$ in Eq.~\eqref{eq:errorbound1}, which we can now write as  
\begin{align}\label{eq:errorbound2}
    \mathbb{P}(N'>N^*) & \leq e^{\log(2) K/2}e^{-(\log(1+\frac{1}{2\sinh^2(r)})N^*/2)}\leq e^{\log(2) K/2}e^{- N^*/(4 \sinh^2(r)+2)}, 
\end{align}
where we have used the inequality $\log(1+x)\geq x/(1+x), \forall x > -1$.  From here, it is possible to see that a choice of 
\begin{equation}\label{eq:N_star}
    N^* \geq (4 \sinh^2(r)+2)\left(\frac{\log(2)}{2}K + c N \log(N)\right), 
\end{equation}
ensures that Eq.~\eqref{eq:errorcondition} is fulfilled. This ``cut-off" in the total photon number can be used to obtain an energy ``cut-off" $J_{\max}$ fulfilling Eq.~\eqref{eq:errorbound2}. According to the Hamiltonian
\begin{equation}
    \hat{H}_{NL} = N^2 \sum_{j=1}^N \hat{n}_j + \sum_{j=1}^N \hat{n}_j^2 + N^2(N+2) \sum_{j=N+1}^M  \hat{n}_j,
\end{equation}
a state of $N'$ photons has a maximum energy $E_{max}(N')= \max (N^2(N+2) N' , N'^2)$. For simplicity, we will assume $N'>N^2(N+2)$, such that $E_{max}(N')= N'^2$. This implies that for $N'>N^2(N+2)$, any eigenstate of $H_{NL}$ with energy larger than $N'^2$ has to have at least $N'$ photons. From this reasoning we conclude that if $\sqrt{J_{max}}>N^2(N+2)$, then
\begin{equation}
  \sum_{j=J_{\max}}^{\infty} p_j\leq \mathbb{P}(N'>\sqrt{J_{max}}).
\end{equation}
Finally, we also assume that the number of photons of the event $N$ of the GBS outcome $S^*$ that we are interested in is large enough such that 
\begin{equation}\label{eq:assumption_Napp}
    N^2(N+2) > (4 \sinh^2(r)+2)\left(\frac{\log(2)}{2}K + c N \log(N)\right). 
\end{equation}
This is fulfilled asymptotically with high probability, given the assumption of the theorem that $\bar{N}/2\leq N\leq 3\bar{N}/2$. To conclude, the assumption~\eqref{eq:assumption_Napp}, together with Eqs.~\eqref{eq:N_star} and~\eqref{eq:errorcondition}, allows us to see that the choice of the cut-off energy $J_{\max}= N^4(N+2)^2= O(N^6)$ ensures that 
\begin{equation}
  \sum_{j=J_{\max}}^{\infty} p_j\leq \mathbb{P}(N'>\sqrt{J_{max}})\leq \exp{(-c N \log(N) )}.
\end{equation}
\end{proof}

\section{Continuous-variable Hadamard test\label{app:CVHadamard}}

Let $\hat{V}$ be a unitary (not necessarily Gaussian) transformation acting on $M$ bosonic modes, that commutes with the total photon number operator, i.e. the sum of the number operators on all $M$ modes. Let $\ket{\Psi_G}$ be a pure Gaussian state of $M$ modes. We suppose that we can construct $\ket{\Psi_G}$, which can always be decomposed as~\cite{Weedbrook2012}
\begin{equation} \label{eq:GaussianPureapp}
    \ket{\Psi_G} = \left( \bigotimes_{m=1}^M \hat{D}_{\alpha_m} \right) \hat{U}_L \left( \bigotimes_{l=1}^M \hat{S}_{r_l} \right) \ket{\boldsymbol{0}},
\end{equation}
where the $\hat{S}_r$ are squeezings, $\hat{U}_L$ is an $M$-mode linear interferometer, the $\hat{D}_{\alpha}$ are displacements and we denote a product of $M$ vacua by $\ket{\boldsymbol{0}} = \ket{0}^{\otimes M}$. We are going to design a CV version of the so-called Hadamard test, that is, a quantum circuit acting on CV states that can be used to estimate the value of the amplitude $\bra{\Psi_G} \hat{V} \ket{\Psi_G}$. In fact, without changing the complexity of our setting, we can deal with a more general case of estimating the amplitude $\bra{\Psi_G'} \hat{V} \ket{\Psi_G}$, where $\ket{\Psi_G'}$ and $\ket{\Psi_G}$ are two different Gaussian states. A scheme of the required operations for this task is presented in Fig.~\ref{fig:CVHadamard}. As we will show, we will require non-Gaussian operations to perform this task, namely the implementation of a controlled version of the Gaussian unitary which prepares $\ket{\Psi_G}$. As summarized in the main text, the procedure involves three main steps which we explain here in detail:~\\~\\
$\boldsymbol{(i)}$~\textbf{Preparing a superposition $\ket{\phi_+} \ket{\boldsymbol{0}} + \ket{\phi_-} \ket{\Psi_G}$.} We start with the state $\left( \bigotimes_{l=1}^M \hat{S}_{-r_l/2} \right) \ket{\boldsymbol{0}}$, and perform a controlled phase operation via the non-Gaussian unitary acting on an ancilla  state and the $M$-mode state  defined as
\begin{equation}
    \hat{U}_{CP}(\varphi) = e^{- i \varphi \hat{n}_0 \sum_{k=1}^{M} \hat{n}_k}. 
\end{equation}
Here, the number operator $\hat{n}_0$ acts on the first ancilla mode, while the $\hat{n}_k$ are the number operators acting on the remaining $M$ modes. This operator can be seen as the controlled application of the phase shift operator
\begin{equation}
\hat{U}_P(j \varphi)= \prod_{k=1}^M\exp(-i j \varphi \hat{n}_k),    
\end{equation}
which introduces a phase-shift of  $j \varphi$ in each of the $M$ modes with the value of $j$ being determined by the number of photons in the ancillary mode. If the ancilla is chosen to be in a coherent state $\ket{\alpha}$, this gives
\begin{equation}
    \hat{U}_{CP}(\varphi) \left( \bigotimes_{l=1}^M \hat{S}_{-r_l/2} \right) \ket{\alpha} \ket{\boldsymbol{0}} = e^{- \frac{1}{2} |\alpha|^2} \sum_{j=0}^{\infty} \frac{\alpha^j}{\sqrt{j!}} \ket{j} \hat{U}_P(j \varphi) \left( \bigotimes_{l=1}^M \hat{S}_{-r_l/2} \right) \ket{\boldsymbol{0}},
\end{equation}
where the $\ket{j}$ are Fock states of the ancilla mode. If we choose $\varphi = \pi/2$ and decompose the state in two components, by splitting the coherent state into the two subnormalized even and odd cat states~\cite{dodonov1974even}
\begin{equation}\label{eq:phi0phi1App}
    \begin{aligned}
        \ket{\phi_+} & = \frac{\ket{\alpha} + \ket{-\alpha}}{2} = e^{- \frac{1}{2} |\alpha|^2} \sum_{j=0}^{\infty} \frac{\alpha^{2j}}{\sqrt{(2j)!}} \ket{2j}, \\
        \ket{\phi_-} & = \frac{\ket{\alpha} - \ket{-\alpha}}{2} = e^{- \frac{1}{2} |\alpha|^2} \sum_{j=0}^{\infty} \frac{\alpha^{2j+1}}{\sqrt{(2j+1)!}} \ket{2j+1},
    \end{aligned}
\end{equation}
we obtain 
\begin{equation}
    \begin{aligned}
        \hat{U}_{CP}(\pi/2) \left( \bigotimes_{l=1}^M \hat{S}_{-r_l/2} \right) \ket{\alpha} \ket{\boldsymbol{0}} = &  \ket{\phi_+} \left( \bigotimes_{l=1}^M \hat{S}_{-r_l/2} \right) \ket{\boldsymbol{0}} + \ket{\phi_-} \left( \bigotimes_{l=1}^M \hat{S}_{r_l/2} \right) \ket{\boldsymbol{0}}.
    \end{aligned}
\end{equation}

\noindent Here we used the fact that the local squeezed vacuum states in the above can be written as a superposition of even Fock states and are therefore invariant under the application of a $\pi$ phase-shift $\hat{U}_P(\pi)$. Additionally, we note that a $\pi/2$ rotation preceding a squeezing of the vaccuum is equivalent to a squeezing in the perpendicular direction in phase space, i.e., anti-squeezing the vaccuum. 
We then perform a squeezing $\left( \bigotimes_{l=1}^M \hat{S}_{r_l/2} \right)$ on the $M$ main modes and obtain the state
\begin{equation}
        \left( \bigotimes_{l=1}^M \hat{S}_{r_l/2} \right) \hat{U}_{CP}(\pi/2) \left( \bigotimes_{l=1}^M \hat{S}_{-r_l/2} \right) \ket{\alpha} \ket{\boldsymbol{0}} = \ket{\phi_+} \ket{\boldsymbol{0}} + \ket{\phi_-} \left( \bigotimes_{l=1}^M \hat{S}_{r_l} \right) \ket{\boldsymbol{0}}.
\end{equation}
The above operation can nicely be interpreted as a squeezing controlled by an effective qubit (the two-level system spanned by the two subnormalized states $\ket{\phi_+}$ and $\ket{\phi_+}$). We now apply the linear interferometer $\hat{U}_L$, before applying a similar controlled displacement. As $\hat{U}_L$
leaves the vaccuum invariant, we have that \begin{equation}
        \hat{U}_L \left( \bigotimes_{l=1}^M \hat{S}_{r_l/2} \right) \hat{U}_{CP}(\pi/2) \left( \bigotimes_{l=1}^M \hat{S}_{-r_l/2} \right) \ket{\alpha} \ket{\boldsymbol{0}} = \ket{\phi_+} \ket{\boldsymbol{0}} + \ket{\phi_-} \hat{U}_L \left( \bigotimes_{l=1}^M \hat{S}_{r_l} \right) \ket{\boldsymbol{0}}.
\end{equation}
Subsequently, we apply the controlled displacement operator which, in analogy to what was previously shown for the controlled squeezing, is divided into three different operations. 
We first apply a displacement $\left( \bigotimes_{m=1}^M \hat{D}_{-\alpha_m/2} \right)$,
\begin{equation}
    \begin{aligned}
         & \left( \bigotimes_{m=1}^M \hat{D}_{-\alpha_m/2} \right) \hat{U}_L \left( \bigotimes_{l=1}^M \hat{S}_{r_l/2} \right) \hat{U}_{CP}(\pi/2) \left( \bigotimes_{l=1}^M \hat{S}_{-r_l/2} \right) \ket{\alpha} \ket{\boldsymbol{0}} \\
        = & \ket{\phi_+} \left( \bigotimes_{m=1}^M \hat{D}_{-\alpha_m/2} \right) \ket{\boldsymbol{0}} + \ket{\phi_-} \left( \bigotimes_{m=1}^M \hat{D}_{-\alpha_m/2} \right) \hat{U}_L \left( \bigotimes_{l=1}^M \hat{S}_{r_l} \right) \ket{\boldsymbol{0}}.
    \end{aligned}
\end{equation}
We now apply another controlled phase shift, this time with a $\pi$ angle, $\hat{U}_{CP}(\pi)$, which gives
\begin{equation}
    \begin{aligned}
         & \hat{U}_{CP}(\pi) \left( \bigotimes_{m=1}^M \hat{D}_{-\alpha_m/2} \right) \hat{U}_L \left( \bigotimes_{l=1}^M \hat{S}_{r_l/2} \right) \hat{U}_{CP}(\pi/2) \left( \bigotimes_{l=1}^M \hat{S}_{-r_l/2} \right) \ket{\alpha} \ket{\boldsymbol{0}} \\
        = & \ket{\phi_+} \left( \bigotimes_{m=1}^M \hat{D}_{-\alpha_m/2} \right) \ket{\boldsymbol{0}} + \ket{\phi_-} \hat{U}_R(\pi) \left( \bigotimes_{m=1}^M \hat{D}_{-\alpha_m/2} \right) \hat{U}_L \left( \bigotimes_{l=1}^M \hat{S}_{r_l} \right) \ket{\boldsymbol{0}} \\
        = & \ket{\phi_+} \left( \bigotimes_{m=1}^M \hat{D}_{-\alpha_m/2} \right) \ket{\boldsymbol{0}} + \ket{\phi_-} \left( \bigotimes_{m=1}^M \hat{D}_{\alpha_m/2} \right) \hat{U}_L \left( \bigotimes_{l=1}^M \hat{S}_{r_l} \right) \ket{\boldsymbol{0}}.
    \end{aligned}
\end{equation}
Finally, we apply another displacement $\left( \bigotimes_{m=1}^M \hat{D}_{\alpha_m/2} \right)$, and get the final state $\ket{\Lambda}$
\begin{equation}
    \begin{aligned}
         \ket{\Lambda} = & \left( \bigotimes_{m=1}^M \hat{D}_{\alpha_m/2} \right) \hat{U}_{CP}(\pi) \left( \bigotimes_{m=1}^M \hat{D}_{-\alpha_m/2} \right) \hat{U}_L \left( \bigotimes_{l=1}^M \hat{S}_{r_l/2} \right) \hat{U}_{CP}(\pi/2) \left( \bigotimes_{l=1}^M \hat{S}_{-r_l/2} \right) \ket{\alpha} \ket{\boldsymbol{0}} \\
        & = \ket{\phi_+} \ket{\boldsymbol{0}} + \ket{\phi_-} \left( \bigotimes_{m=1}^M \hat{D}_{\alpha_m} \right) \hat{U}_L \left( \bigotimes_{l=1}^M \hat{S}_{r_l} \right) \ket{\boldsymbol{0}}
    \end{aligned}
\end{equation}
or,
\begin{equation} \label{eq:LambdaApp}
    \ket{\Lambda} = \ket{\phi_+} \ket{\boldsymbol{0}} + \ket{\phi_-} \ket{\Psi_G},
\end{equation}
where $\ket{\Psi_G}$ is the pure Gaussian state defined in Eq.~\eqref{eq:GaussianPureapp}.~\\~\\
$\boldsymbol{(ii)}$~\textbf{Applying the unitary $\hat{V}$.} Since the unitary $\hat{V}$ does not affect a product of vacua, applying it on the $M$ main modes gives
\begin{equation}
    \hat{V} \ket{\Lambda} = \ket{\phi_+} \ket{\boldsymbol{0}} + \ket{\phi_-} \hat{V} \ket{\Psi_G}.
\end{equation}
$\boldsymbol{(iii)}$~\textbf{Measuring a probability.} By performing a Gaussian measurement, we can measure the outcome probability
\begin{equation}
    \begin{aligned}
         |\bra{\alpha} \bra{\Psi'_G} \hat{V} \ket{\Lambda}|^2 = | \bk{\alpha}{\phi_+} \bk{\Psi'_G}{\boldsymbol{0}} + \bk{\alpha}{\phi_-} \bra{\Psi'_G} \hat{V} \ket{\Psi_G}|^2.
    \end{aligned}
\end{equation}
Now,
\begin{equation}
    \bk{\alpha}{\phi_+} = e^{- |\alpha|^2} \sum_{j=0}^{\infty} \frac{|\alpha|^{4j}}{(2j)!} = e^{- |\alpha|^2} \cosh(|\alpha|^2),
\end{equation}
and
\begin{equation}
    \bk{\alpha}{\phi_-} = e^{- |\alpha|^2} \sum_{j=0}^{\infty} \frac{|\alpha|^{2(2j+1)}}{(2j+1)!} = e^{- |\alpha|^2} \sinh(|\alpha|^2),
\end{equation}
so that
\begin{equation}
    \bra{\alpha} \bra{\Psi'_G} \hat{V} \ket{\Lambda} = e^{- |\alpha|^2} \left( \cosh(|\alpha|^2) \bk{\Psi'_G}{\boldsymbol{0}} + \sinh(|\alpha|^2) \bra{\Psi'_G} \hat{V} \ket{\Psi_G} \right).
\end{equation}
Finally, we can express the outcome probability as
\begin{equation}\label{eq:final_probability_appCVHad}
    \begin{aligned}
        |\bra{\alpha} \bra{\Psi'_G} \hat{V} \ket{\Lambda}|^2 & = e^{-2|\alpha|^2} \Big( \cosh^2(|\alpha|^2) \left| \bk{\Psi'_G}{\boldsymbol{0}} \right|^2 + \sinh^2(|\alpha|^2) \left| \bra{\Psi'_G} \hat{V} \ket{\Psi_G} \right|^2 \\
        & \qquad \qquad \qquad + 2 \cosh(|\alpha|^2) \sinh(|\alpha|^2) \, \mathrm{Re}\left[ \bk{\boldsymbol{0}}{\Psi'_G} \bra{\Psi'_G} \hat{V} \ket{\Psi_G} \right]  \Big).
    \end{aligned}
\end{equation}
This equation defines the constants $c_1$, $c_2$ and $c_3$ from Eq.~(20) of the main text, which are functions of the parameter $\alpha$. 

We note that we could have post-selected the ancilla on the coherent state $\ket{\alpha}$ rotated by an angle $\pi/2$ instead, which would have given
\begin{equation}
    \bra{\alpha} e^{i \frac{\pi}{2} \hat{n}_0} \ket{\phi_+} = e^{- |\alpha|^2} \sum_{j=0}^{\infty} (-1)^j \frac{|\alpha|^{4j}}{(2j)!} = e^{- |\alpha|^2} \cos(|\alpha|^2),
\end{equation}
and
\begin{equation}
    \bra{\alpha} e^{i \frac{\pi}{2} \hat{n}_0} \ket{\phi_-} = i e^{- |\alpha|^2} \sum_{j=0}^{\infty} (-1)^j \frac{|\alpha|^{2(2j+1)}}{(2j+1)!} = i e^{- |\alpha|^2} \sin(|\alpha|^2),
\end{equation}
so that
\begin{equation}
    \bra{\alpha} e^{i \frac{\pi}{2} \hat{n}_0} \bra{\Psi'_G} \hat{V} \ket{\Lambda} = e^{- |\alpha|^2} \left( \cos(|\alpha|^2) \bk{\Psi'_G}{\boldsymbol{0}} + i \sin(|\alpha|^2) \bra{\Psi'_G} \hat{V} \ket{\Psi_G} \right),
\end{equation}
and
\begin{equation}
    \begin{aligned}
        |\bra{\alpha} e^{i \frac{\pi}{2} \hat{n}_0} \bra{\Psi'_G} \hat{V} \ket{\Lambda}|^2 & = e^{-2|\alpha|^2} \Big( \cos^2(|\alpha|^2) \left| \bk{\Psi'_G}{\boldsymbol{0}} \right|^2 + \sin^2(|\alpha|^2) \left| \bra{\Psi'_G} \hat{V} \ket{\Psi_G} \right|^2 \\
        & \qquad \qquad \qquad - 2 \cos(|\alpha|^2) \sin(|\alpha|^2) \, \mathrm{Im}\left[ \bk{\boldsymbol{0}}{\Psi'_G} \bra{\Psi'_G} \hat{V} \ket{\Psi_G} \right]  \Big).
    \end{aligned}
\end{equation}
To summarize, by defining $c=\cos(|\alpha|^2)$, $s=\sin(|\alpha|^2)$, $ch=\cosh(|\alpha|^2)$ and $sh=\sinh(|\alpha|^2)$, we end up with
\begin{equation} \label{eq:c1c2c3}
    \mathrm{Re}\left[ \bk{\boldsymbol{0}}{\Psi'_G} \bra{\Psi'_G} \hat{V} \ket{\Psi_G} \right] = \frac{e^{2|\alpha|^2}}{2 \, ch \, sh} |\bra{\alpha} \bra{\Psi'_G} \hat{V} \ket{\Lambda}|^2 - \frac{ch}{2 \, sh} \left| \bk{\Psi'_G}{\boldsymbol{0}} \right|^2 - \frac{sh}{2 \, ch} \left| \bra{\Psi'_G} \hat{V} \ket{\Psi_G} \right|^2,
\end{equation}
and
\begin{equation}
        \mathrm{Im}\left[ \bk{\boldsymbol{0}}{\Psi'_G} \bra{\Psi'_G} \hat{V} \ket{\Psi_G} \right] = \frac{c}{2 \, s} \left| \bk{\Psi'_G}{\boldsymbol{0}} \right|^2 + \frac{s}{2 \, c} \left| \bra{\Psi'_G} \hat{V} \ket{\Psi_G} \right|^2 - \frac{e^{2|\alpha|^2}}{2 \, c \, s} |\bra{\alpha} e^{i \frac{\pi}{2} \hat{n}_0} \bra{\Psi'_G} \hat{V} \ket{\Lambda}|^2.
\end{equation}
Since we assume we know the state $\ket{\Psi'_G}$, the overlap $\bk{\boldsymbol{0}}{\Psi'_G}$ can be computed beforehand. Also, we assume that the probability $|\bra{\Psi'_G} \hat{V} \ket{\Psi_G}|^2$ is estimated beforehand. Thus, the last two equations allows us to retrieve  the complex-value amplitude $\bra{\Psi'_G} \hat{V} \ket{\Psi_G}$ from the probability measured at step (iii) of the proposed CV Hadamard test.

\end{document}